\documentclass[10pt,journal,compsoc]{IEEEtran}

\usepackage{graphicx}
\usepackage{subfigure}
\usepackage{algpseudocode}
\usepackage{amsmath}
\usepackage{amsfonts,amssymb}
\usepackage{amsthm}
\usepackage{multirow}
\usepackage[ruled]{algorithm2e}

\newtheorem{theorem}{Theorem}
\newcommand{\tabincell}[2]{\begin{tabular}{@{}#1@{}}#2\end{tabular}}

\begin{document}

\title{PPtaxi: Non-stop Package Delivery via Multi-hop Ridesharing}

\author{Yueyue Chen, Deke Guo,~\IEEEmembership{Senior Member,~IEEE,}
        Ming Xu,~\IEEEmembership{Member,~IEEE,}
        Guoming Tang, Tongqing Zhou, and Bangbang Ren % <-this % stops a space
\IEEEcompsocitemizethanks{\IEEEcompsocthanksitem Y Chen, M Xu, and T Zhou are with the college of computer, National University of Defense Technology, Changsha, China. D Guo is with the College of System Engineering, National University of Defense Technology, Changsha, China; he is also with the School of Computer Science and Technology, Tianjin University, Tianjin, China. G Tang, and B Ren are with the College of System Engineering, National University of Defense Technology, Changsha, China. \protect\\Email: \{yueyuechen, dekeguo, xuming, gmtang, zhoutongqing, bangbangren\}@nudt.edu.cn}}% <-this % stops an unwanted space

%{Shell \MakeLowercase{\textit{et al.}}: Bare Demo of IEEEtran.cls for Computer Society Journals}


\IEEEtitleabstractindextext{%
\begin{abstract}
City-wide package delivery becomes popular due to the dramatic rise of online shopping. It places a tremendous burden on the traditional logistics industry, which relies on dedicated couriers and is labor-intensive. Leveraging the ridesharing systems is a promising alternative, yet existing solutions are limited to one-hop ridesharing or need consignment warehouses as relays. In this paper, we propose a new package delivery scheme which takes advantage of multi-hop ridesharing and is entirely consignment free. Specifically, a package is assigned to a taxi which is guided to deliver the package all along to its destination while transporting successive passengers. We tackle it with a two-phase solution, named \textbf{PPtaxi}. In the first phase, we use the Multivariate Gauss distribution and Bayesian inference to predict the passenger orders. In the second phase, both the computation efficiency and solution effectiveness are considered to plan package delivery routes. We evaluate \textbf{PPtaxi} with a real-world dataset from an online taxi-taking platform and compare it with multiple benchmarks. The results show that the successful delivery rate of packages with our solution can reach $95\%$ on average during the daytime, and is at most $46.9\%$ higher than those of the benchmarks.
\end{abstract}

}

\maketitle

\IEEEdisplaynontitleabstractindextext

% \fi
\IEEEpeerreviewmaketitle

%--------------------------------------------------------------------------------------------------------------------------------------
\section{Introduction}

Along with the pervasive of mobile payment and the rapid development of the logistics industry, online shopping has become popular nowadays and has brought great convenience to daily lives~\cite{lowe2014last}. To further improve the shopping experience of customers, many online retailers (like Amazon~\cite{amazon} and JD~\cite{jd}) begin to offer the city-wide \emph{same-day delivery solution} (i.e., the package is promised to be delivered to the customer within 24 hours). Traditional approaches rely on dedicated couriers for package delivery. Wherein, speeding up the delivery process usually means to put in additional personnel and vehicles, which adds up to a higher delivery cost. Therefore, how to reduce the delivery cost while controlling the delivery time is significantly important for the online retailers, the logistics providers, and the customers.

%\subsection{City-wide Package Delivery via Ridesharing}

A promising alternative to the traditional package delivery methods is by leveraging the ridesharing systems, as the potential delivery capability through ridesharing within an urban environment is enormous \cite{wangridesharing}. By allowing the passengers and the packages with similar itineraries and time schedules to share one vehicle, the delivery cost can be reduced significantly \cite{furuhata2013ridesharing,rouges2014crowdsourcing}. 

Nevertheless, existing ridesharing solutions for package delivery are with two major defects, stopping them from being widely adopted in practice. First, most ridesharing solutions attempt to deliver a package in one shot, i.e., one-hop ridesharing \cite{Arslan2016,wangridesharing,Liu2018}. For example, Walmart proposes to make use of its in-store customers to deliver goods to its online customers on their way home from the store \cite{Arslan2016}. Wang et al. \cite{wangridesharing} encourage a private car to change its regular route to a similar one which passes through the package pick up and drop off locations. Such one-hop ridesharing systems provide little chance to be utilized in the city-wide package delivery, as the similar itineraries with likely origin and destination are quite limited. Second, for those packages cannot be delivered within one hop ridesharing, consignment warehouses are needed as relays. For instance, Chen et al. \cite{Chen2017} propose to place consignment warehouses along roads as relays for taxis to support the ridesharing system. Thus, the delivery process of one package needs the cooperation of several taxis, and the taxi drivers have to get on and off frequently to pick up and drop down the packages. With such a method, the storage cost is increased dramatically as a large number of warehouses need to be built to cover a whole city.

In this paper, we propose a new package delivery scheme, which uses ridesharing of crowdsourced taxis while eliminating the aforementioned two defects. With our scheme, as shown in Fig. \ref{fig_sce}, a package is assigned to a selected taxi, which is requested to deliver the package all along. Meanwhile, the taxi is also able to transport one or more passengers successively until the package reaches its destination. During the whole process of package delivery, i) the benefit of multi-hop ridesharing is well explored and ii) no intermediate consignment warehouses are needed (i,e., non-stop package delivery). 

\begin{figure}
\centering
\includegraphics[width=3in]{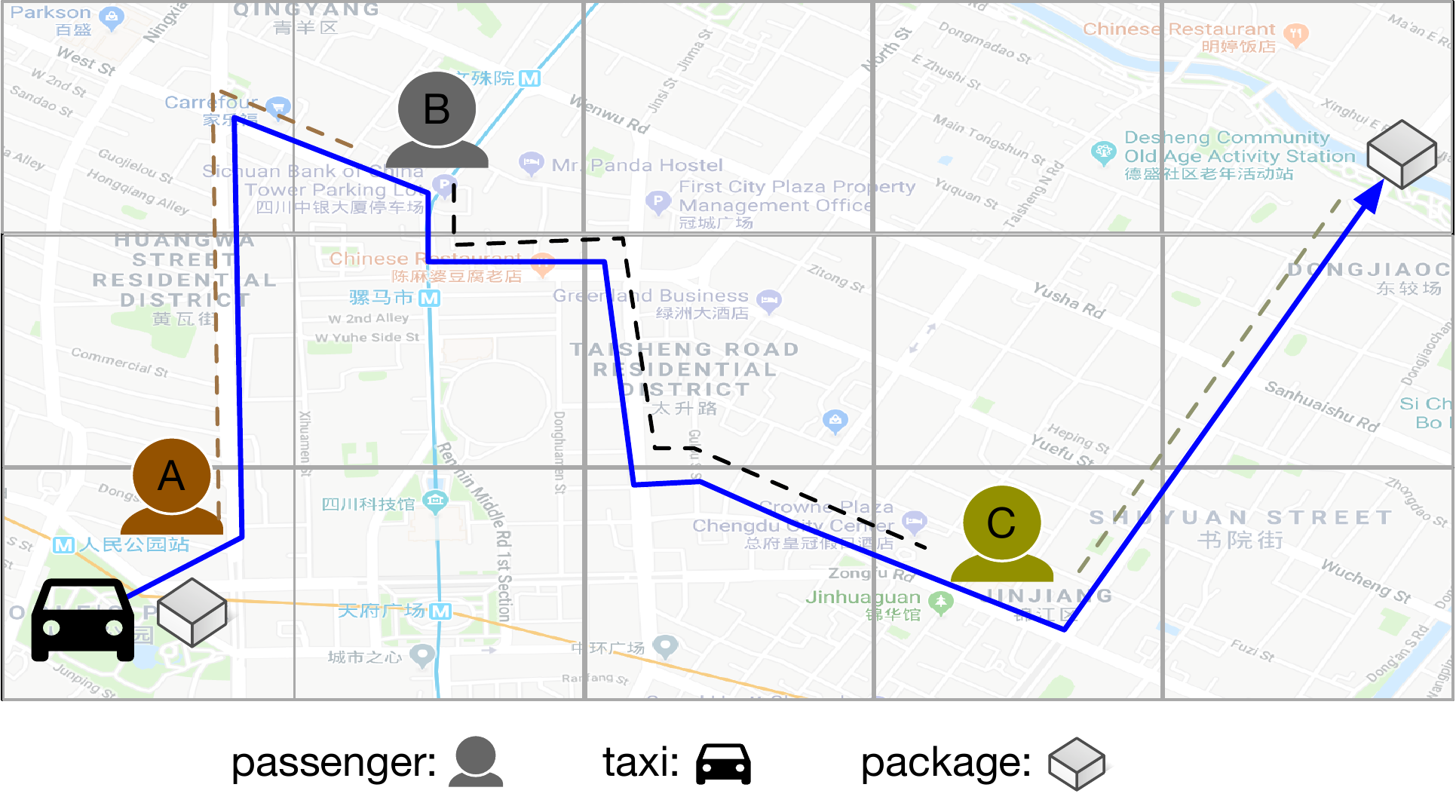}
\caption{An illustration of our package delivery scheme. The dotted lines denote the paths of passengers, and the solid line denotes the path of a package. By strategically taking the passenger orders (A, B, and C in the figure), the taxi successfully delivers the package to its destination.} 
\label{fig_sce}
\end{figure}

%\subsection{Challenges of Non-stop Package Delivery via Multi-hop Ridesharing}

During the non-stop package delivery via multi-hop ridesharing, a package should follow a route forward to its destination while consecutive passenger orders are popping up along the route. One main difficulty of planning such a route is that future passenger orders are unknown. Without considering the demands of passenger orders, a taxi tends to choose the nearest passenger to pick up \cite{tong2016}. Following such a prediction-oblivious route, there may be little chance to deliver the package successfully.

Consequently, there are two main challenges to implement our non-stop package delivery via multi-hop passengers ridesharing. \textbf{Q1: How to predict the city-wide demands of passenger orders?} The passenger travel demands usually show weekly and daily pattern \cite{lin2018optimal,Zhang2017taxi}. Most demands prediction work focuses on estimating the order number for a given location and time \cite{tong2017simpler,yao2018deep}, while few of them study the passenger flows \cite{Zhang2017taxi}. In our context, however, the city-wide passenger flows in the future are required to find the (potential) best package delivery routes. \textbf{Q2: How to plan package delivery routes with dynamic passenger orders?} Two kinds of requirements should be met when planning delivery routes: i) the route selection should be completed in real time since taxis move fast and passenger order distribution varies over time; ii) the ratio of successfully delivered packages should be high enough by utilizing the selected delivery routes. Otherwise, the package delivery scheme is problematic in practice.

%\subsection{Our Contributions}

By addressing the above challenges, we make the following contributions in this paper. 

\begin{itemize}
\item We propose and formulate the non-stop package delivery problem (NPD) of route probability maximization, subject to package delivery delay constraint (Section \ref{sec_prob}). We prove that the problem is NP-Complete.
\item We design a two-phase solution to the NPD problem, named \textbf{PPtaxi} (Section \ref{sec_solution}). To solve the passenger prediction problem (Q1), we use the Multivariate Gauss distribution and the Bayesian inference to calculate the probability of each passenger flow. To solve the route planning problem (Q2) and meet the two requirements, we propose a three-stage strategy, in which a Dijkstra-based optimal route planning algorithm, a prediction based sequential route planning algorithm and a heuristic sequential planning algorithm are developed, respectively.
\item We evaluate \textbf{PPtaxi} with a real-world dataset from an online taxi-taking platform and compare its performance with multiple benchmarks (Section \ref{sec_experiment}). The results show the effectiveness and efficiency of our solution.
\end{itemize}

The advantages of applying our solution are three-fold: i) for the taxi driver: the taxi's free space is further utilized; ii) for the online retailer and logistics provider: no consignment warehouse (i.e., storage cost) is needed anymore between the package origin and destination; iii) for the online shopping customer: the corresponding cost of the package delivery is decreased. 

In addition, Section \ref{sec_relat} reviews previous work related to this paper. We conclude the paper in Section \ref{sec_conclusion}.

%--------------------------------------------------------------------------------------------------------------------------------------
\section{Problem Statement}
\label{sec_prob}

In this section, we first introduce the scenario of our non-stop package delivery scheme, then formulate the system model and objective function, at last, we analyze the problem hardness.

%--------------------------------------------------------------------------
\subsection{Non-stop Package Delivery Scenario}

We illustrate a scenario example of our non-stop package delivery scheme in Fig. \ref{fig_sce}. Suppose a taxi has already picked up a package in origin, and needs to transport the package to the destination. The taxi can transport A, B, C three passengers one after another while carrying the package. The drop off location and time of passenger A (B) are close to the pickup location and time of passenger B (C), the destination of passenger C is close to the package destination. After dropping off the passenger C, the taxi will deliver the package to its destination directly. In our solution, there is a cloud platform to schedule taxis and dispatch passenger orders. Note that, the matching problem between taxis and packages, which has been researched widely in recent papers \cite{tong2016,tong2017flexible}, is not discussed in this paper. As the route of each taxi is determined by the dispatched passenger orders, which taxi to select to fetch the package is not the key impactor in our problem. In this paper, a package is defaulted to be fetched by the nearest taxi. 

To achieve our solution in the real world, we make some practical assumptions in this paper. 

\emph{Assumption 1: The information of passengers is not known to the platform until they make taxi-taking requests on the platform; the taxi drivers are requested to accept the passenger orders assigned by the platform.}

In fact, this assumption is already achieved on the online taxi-booking platform, such as DiDi\cite{didi} and Uber\cite{uber}. Once a passenger launches a taxi-taking request on the platform, he/she must upload the related information to match proper taxis. The passenger orders are assigned to taxis by the platform, and the taxi drivers who are willing to join the platform are requested to accept the assignment unless with special cases.

\emph{Assumption 2: The taxi drivers are requested to deliver packages if they are selected. There is always spare space for the delivered package in the selected taxi, and the package delivery does not bring inconvenience to passengers.}

With our non-stop package delivery scheme, the delivery cost of a package can be reduced dramatically compared with traditional delivery methods. Thus the delivery payments from customers can be spared to design proper incentive mechanisms for taxi drivers. We believe this assumption can be realistic given proper incentive mechanism \cite{gao2015survey,zhang2016incentives}. As it is beyond the scope of this paper, we will not discuss the incentive mechanism.

\emph{Assumption 3: The traces of taxis are all trackable to make sure the security of the delivered packages.}

The package security is an important problem in our solution. Fortunately, there is a relatively comprehensive driver profiling mechanism in today's platform. By utilizing our solution, a taxi is requested to deliver a package all along. Combined with the on-record taxi trajectories in the third assumption and driver profiling mechanism, the package could be traceable, and liability cognizance is possible to be done if an accident happens.

%--------------------------------------------------------------------------
\subsection{System Model}

\textbf{Network model:} We construct a block-based graph for the target area, as illustrated in Fig. \ref{fig_model}. We model it as a directed graph $G_0(\beta_0, \varepsilon_0)$, where $\beta_0$ is the node set and $\varepsilon_0$ is the edge set. More specifically, we divide the target area into $M$ blocks. For each block, we assign a representative node from which all other nodes in this block can be reached in a given time, and let the node represents the block. To simplify the parameters, we denote $b_i$ as the representative node of the $i^{th}$ block in our following modeling and analysis. For each edge $e \mathrm{\in} \varepsilon_0$, we allocate two weight, $P(b_j, b_i)$ is denoted as the predicted passenger-demand probability and $\delta(b_i,b_j)$ as the pass time of the shortest path between node $b_i$ and $b_j$. Note that, the probability weight $P$ of each edge varies according to different departure time $t$ in our model. Specifically, in block $b_i$ at time $t$, there is a passenger demand aiming to block $b_j$ with probability $P(b_j, b_i | t)$. The detailed calculation of the passenger demand prediction can be seen in Section \ref{sec_order}. 

\begin{figure}
\centering
\includegraphics[width=3.3in]{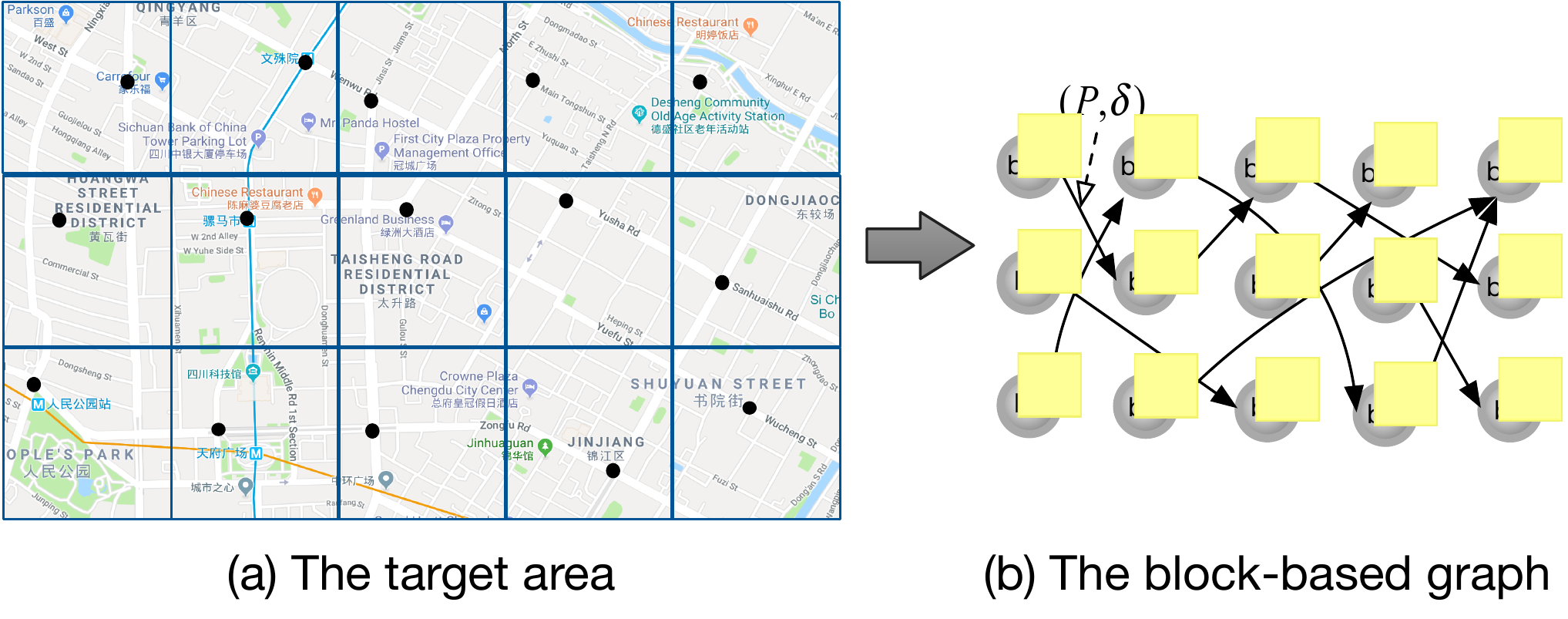}
\caption{An illustration of the network model. The target area is divided into multiple blocks, and each block has a representative node marked as a black dot. Each node in the constructed block-based graph represents a block in the target area. Each edge has two weights, i.e., the predicted passenger probability $P$ and the shortest path time $\delta$ between nodes.}
\label{fig_model}
\end{figure}

\textbf{Package request model:} Suppose the package set is denoted as $U$. For a package $u$ ($u\mathrm{\in}U$), the departure location, the destination, the departure time\footnote{The departure time of a package is denoted as the departure time of the matched taxi, which is different with the generation time (denoted as $u_{genT}$), in Section \ref{sec_experiment}, we will show the relationship between the generation time and the departure time of a package.}, and the delivery time of a package request $u$ are denoted as $u_{dep}$, $u_{des}$, $u_{depT}$, and $u_{desT}$, respectively. Note that, the package requests appear dynamically to the platform, the former three parameters of a package can be known as the package request appears. The platform needs to determine a route for a package request from $u_{dep}$ to $u_{des}$, and the route consists of several passenger paths. The delivery time of the package $u_{desT}$ depends on the results of the selected route.

\textbf{Online passenger model:} The information about a passenger order, including the departure location, departure time, destination, and delivery time are denoted as $v_{dep}$, $v_{depT}$, $v_{des}$ and $v_{desT}$, respectively. Note that, the passengers must be delivered in time. Hence, a passenger path is always the shortest path between $v_{dep}$ and $v_{des}$. Since passengers appear dynamically to the platform, the future passenger orders are not known.

%--------------------------------------------------------------------------
\subsection{The NPD Problem Formulation}

The goal of this paper is to deliver packages successfully within given deadlines while reducing the extra cost incurred by the express speeding up. To achieve this goal, we propose the non-stop package delivery scheme, i.e., utilize one taxi to deliver a package all along without stop through passengers ridesharing. The taxi can delivery multiple passengers sequentially while carrying the package, and the package delivery route is composed of one or more passenger paths. In this solution, we take the passenger demands prediction into account when planning a package delivery route. Based on historical order data, we mine the potential regulars of the passenger demands and calculate the probability of different passenger flows. According to the passenger prediction results, we further select the route with the maximum probability for each package. 

Suppose a possible route for package $u$ is denoted as $r\textrm{=}(b_1, b_2, \dots, b_m, b_{m+1})$, where $b_1 \textrm{=} u_{dep}$ and $b_{m+1}\textrm{=}u_{des}$. There are $m$ passenger paths in this route, as $b_i\textrm{=}v^{(i)}_{dep}$, $b_j\textrm{=}v^{(i)}_{des}\textrm{=}v^{(i+1)}_{dep}$, $v^{(i)}$ denotes the $i^{th}$ passenger path in the planned route. Note that, the intermediate nodes within a passenger path is not listed in the package route, as we only concern the departure and destination information about the passenger path. Then along this route, the package will reach $b_i$ at time $t_{b_i}$, as $t_{b_1}\textrm{=}u_{depT}$, thus $t_{b_i}\mathrm{\buildrel \Delta \over = }u_{depT}\mathrm{+}\delta(b_1,b_2)\mathrm{+}\delta(b_2,b_3)\mathrm{+}\dots\mathrm{+}\delta(b_{i-1},b_i)$. Therefore, the selected route consists of consecutive passenger orders, where the departure time of the latter passenger is the destination time of the former passenger. Thus, the passenger probabilities along a route are independent of each other. The probability of this route can be calculated as: 
\begin{equation}
P(r)\buildrel \Delta \over = \prod\limits_{i\mathrm{\in}[1,m], j\mathrm{\in}[2,m+1]} {P({b_j},{b_i}|{t_{{b_i}}})}
\end{equation}

\begin{figure*}[t]
\centering
\includegraphics[width=5in]{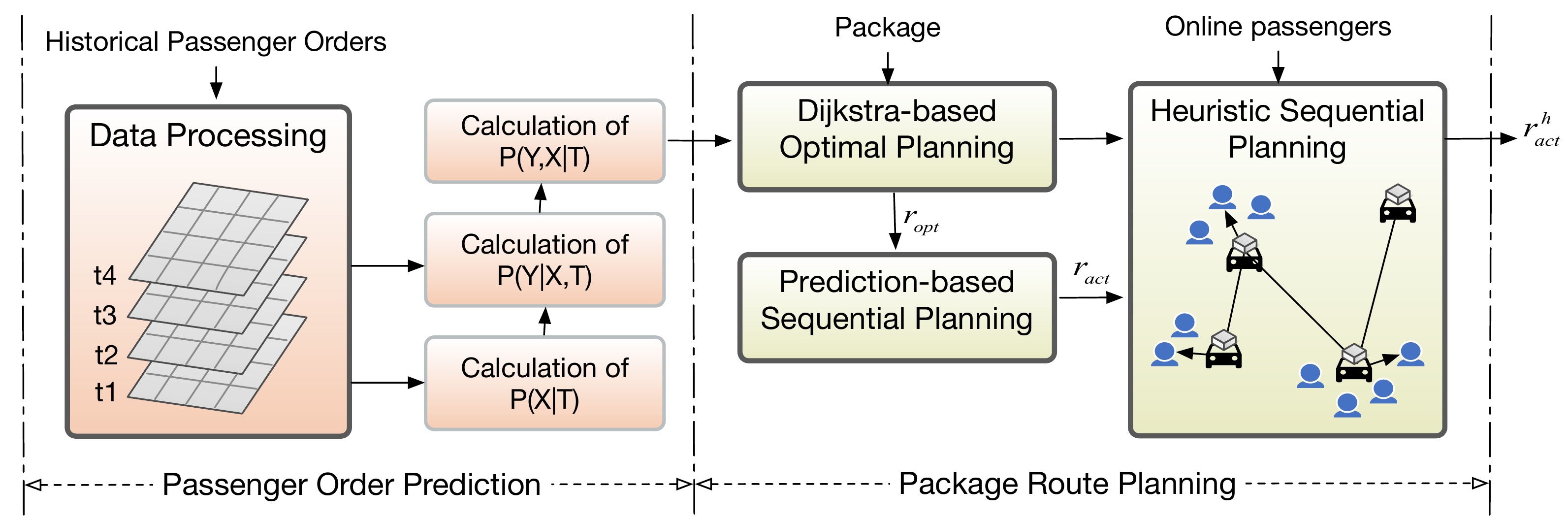}
\caption{The framework of PPtaxi.}
\label{fig_framework}
\end{figure*}

Then our NPD problem can be formulated as: Given a package request $u_{dep}$, $u_{des}$, $u_{depT}$, and the online passenger orders $V$, find the route with the maximum probability, under the delivery time constraint.
\begin{equation}
\label{equ_pr}
\begin{array}{l}
\mathop {\max }\limits_{r \in R} \ \ P(r)\\
s.t.\ \ {u_{desT}} \mathrm{-} {u_{depT}} \mathrm{\le} max T
\end{array}
\end{equation}
where $R$ is the set of all possible routes from $u_{dep}$ to $u_{des}$ with the departure time of $u_{depT}$, the constraint restricts that the package must be delivered with a given time length (for example, $24$ hours). Note that the total number of passenger paths $m$, the node$\_b_i$ departure time $t_{b_i}$, and the arrival time of the package $u_{desT}$ are depend on the selected route $r$.

%--------------------------------------------------------------------------
\subsection{Problem Hardness}

According to our problem formulation, the goal of our NPD problem is to maximize the probability product of the multiple passenger paths in the selected route, and the constraint is the the package delivery delay. We can have the following theorem:

\begin{theorem}
The NPD problem is NP-Complete.
\end{theorem}
\begin{proof}
The NPD problem goal (Eq. (\ref{equ_obj})) can be transformed into the addition formulation (Eq. (\ref{equ_sum})) by a move $log(P)$ (Eq. (\ref{equ_log})). As $P(b_j,b_i|t)\mathrm{<}1$, thus $log(P(b_j,b_i|t))\mathrm{<}0$. Therefore, our goal can be further transformed into Eq. (\ref{equ_min}).
\begin{eqnarray}
\label{equ_obj}
&\max& \prod\limits_{i \in [1,m],j \in [2,m + 1]} {P({b_j},{b_i}|{t_{{b_i}}})} \\
\label{equ_log}
 \Leftrightarrow &\max& \log (\prod\limits_{i \in [1,m],j \in [2,m + 1]} {P({b_j},{b_i}|{t_{{b_i}}})} )\\
\label{equ_sum}
 \Leftrightarrow &\max& \sum\limits_{i \in [1,m],j \in [2,m + 1]} {\log (} P({b_j},{b_i}|{t_{{b_i}}}))\\
\label{equ_min}
 \Leftrightarrow &\min& \sum\limits_{i \in [1,m],j \in [2,m + 1]} {( - \log (} P({b_j},{b_i}|{t_{{b_i}}})))
\end{eqnarray}

We redefine the passenger-demand-probability weight of the edge in the network model as $-log(P(b_j, b_i | t))$. After the analysis and $-log(P(b_j, b_i | t))\mathrm{>}0$, the goal of our problem can be transformed into finding the shortest path under the constraint (${u_{desT}} \mathrm{-} {u_{depT}} \mathrm{\le} max T$). It turns out that given the predicted probabilities, the NPD problem can be regarded as a case of the well-known \emph{Restricted Shortest Path (RSP) problem} \cite{biswas2015restricted,lewis1983computers}. In RSP, a directed graph is given where each edge has a fixed travel time and travel cost, and the goal is to find a minimal-cost path subject to a hard path delay requirement. Since RSP is NP-Complete \cite{lewis1983computers}, we can thus easily prove that our problem NPD is also NP-Complete.
\end{proof}

%--------------------------------------------------------------------------------------------------------------------------------------
\section{\textbf{PPtaxi}: A Solution To the NPD Problem}
\label{sec_solution}

As analyzed above, the goal of the NPD problem is to select the route with the maximum probability for each package request based on the passenger demands prediction. As the basement, we first should calculate the order probabilities of the future passengers across the city-wide at different time slots. The intuition to predict the passenger orders is to mining the historical passenger data and find potential regularities. Based on the prediction results, the route with the maximum probability should be selected via proper algorithms. As the NPD problem is NP-Complete, it is impossible to obtain the exact solution in polynomial time unless P = NP. Thus, approximate algorithms need to be designed to find the proper route for each package. 

According to the analysis, we design a solution as shown in Fig. \ref{fig_framework}. The solution is named Passenger \& Package in one taxi, i.e., \textbf{PPtaxi}. There are two main phases in our framework, i) passenger order prediction and ii) package route planning. 

In the phase-i, we utilize the historical data of passenger orders to mine the potential probability of passenger flows across the city at different time slots. More specifically, we first process the historical data and formulate the data into our models, then calculate the departure probability under a given departure time $P(X|T)$, the destination probability under the given departure location and departure time $P(Y|X,T)$, and the flow probability $P(Y,X|T)$, respectively. The detailed design of the order prediction phase is introduced in Section~\ref{sec_order}. 

In the phase-ii, we propose a three-stage mechanism to plan the package delivery route. First, we design a Dijkstra-based optimal route planning algorithm (DOP) based on the predicted passenger orders. The DOP can be achieved in Pseudo-polynomial time, which proves that the NPD problem is NP-Hard in the weak sense. Second, we find that as it is impossible to predict the city-wide future passenger orders with $100\%$ accuracy, the mismatching phenomenon between prediction results and dynamic passenger orders will always exist. To tackle the mismatching problem, we further propose a prediction-based sequential route planning algorithm (PSP). At last, due to the high computation complexity of the PSP algorithm, we further propose a heuristic sequential planning algorithm (HSP) to make a tradeoff between the computation efficiency and the solution effectiveness. The detailed illustration of the second phase is introduced in Section~\ref{sec_pkg}.

%--------------------------------------------------------------------------------------------------------------------------------------
\subsection{Phase-i: Passenger Order Prediction}
\label{sec_order}

In this subsection, we aim to model the probability distribution of passenger demands by mining the historical data of orders. As studied in many papers, the passenger travel demands usually show weekly and daily pattern \cite{lin2018optimal,Zhang2017taxi}. Thus, we classify days into weekdays and weekends. Although we cannot know the future orders within a day exactly, the possible demand of a bunch of passengers shows potential regularity. For example, the residents within an apartment complexes trends to go to work among $8\mathrm{:}00 \mathrm{-}8\mathrm{:}30$ and come home among $5\mathrm{:}30 \mathrm{-} 6\mathrm{:}00$ on weekdays. Suppose the target area is divided into $M$ blocks, and one day is split into $N$ slots. We can formulate the order probability from block $b_i$ to $b_j$ departing at time $t_k$ ($1 \mathrm{\le} i,j  \mathrm{\le} M$, $ 1 \mathrm{\le} k \mathrm{\le} N $) as $P(Y \textrm{=} {b_j},X \textrm{=} {b_i}|T \textrm{=} {t_k})$, where $X$ denotes the departure location, $Y$ denotes the destination, and $T$ denotes the departure time. 

In the conditional probability$P(Y,X|T)$, the joint distribution of two parameters ($Y, X$) under different time slots need to be discussed. According to the multiplication formula of the conditional probability, which is known as $P(Y,X,T) \textrm{=} P(Y|X,T) \mathrm{\times} P(X|T) \mathrm{\times} P(T)$, and $P(Y,X|T) \textrm{=} P(Y,X,T)/P(T)$, we can have the following equations:
\begin{equation}
\label{equ1}
\begin{array}{l}
P(Y \textrm{=} {b_j},X \textrm{=} {b_i}|T \textrm{=} {t_k})\\
 \textrm{=} P(Y \textrm{=} {b_j}|X \textrm{=} {b_i},T \textrm{=} {t_k}) \mathrm{\times} P(X \textrm{=} {b_i}|T \textrm{=} {t_k})
\end{array}
\end{equation}

Eq. (\ref{equ1}) indicates that, at slot $t_k$ and at block $b_i$, there is a travel demand with probability $P(X \textrm{=} {b_i}|T \textrm{=} {t_k})$, and the passengers goes to block $b_j$ with probability $P(Y \textrm{=} {b_j}|X \textrm{=} {b_i},T \textrm{=} {t_k})$. In the following content, we will discuss the conditional departure probability $P(X|T)$ and the conditional destination probability $P(Y|X,T)$, respectively.

%--------------------------------------------------------------------------
\subsubsection{Calculation of Conditional Departure Probability}

In this subsection, we aim to model the probability distribution of each departure block. To achieve this goal, we utilize Bayesian formula\footnote{At least one order is existed in each slot in our experiment, which can be seen in Section \ref{sec_experiment}} to express the conditional probability of the departure blocks $P(X|T)$ as:
\begin{equation}
\label{equ3}
P(X \textrm{=} {b_i}|T \textrm{=} {t_k}) \textrm{=} \frac{{P(T \textrm{=} {t_k}|X \textrm{=} {b_i}) \mathrm{\times} P(X \textrm{=} {b_i})}}{{\sum_i {P(T \textrm{=} {t_k}|X \textrm{=} {b_i}) \mathrm{\times} P(X \textrm{=} {b_i})} }}
\end{equation}
There are two parameters need to be calculated in this formulation, $P(T|X)$ and $P(X)$.

\begin{figure}
\centering
\includegraphics[width=2.2in]{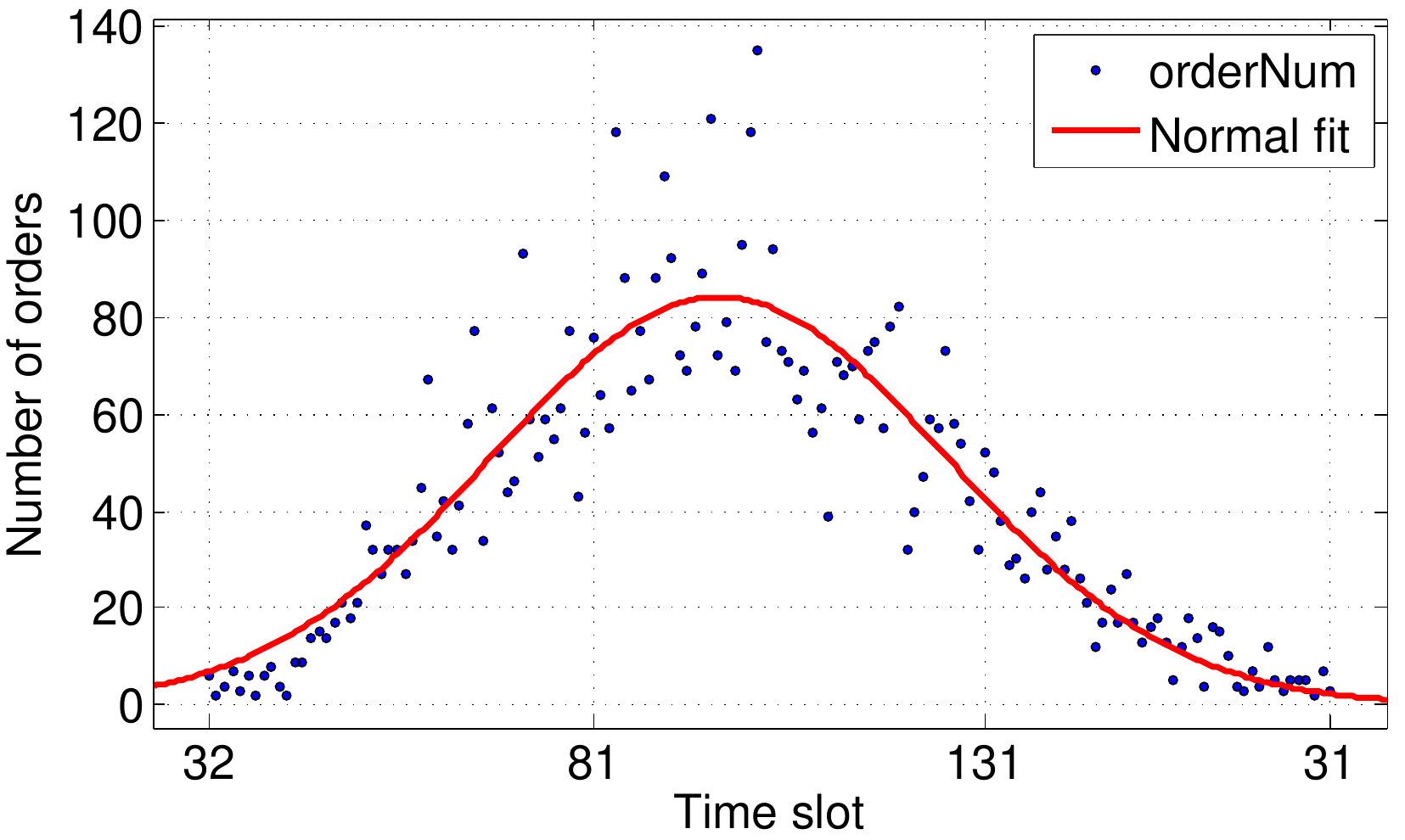}
\caption{An illustration of the passenger order numbers in different time slots within a given location and its normal fit. The scatterplot is the number of passenger orders under different departure time in historical data. The solid line is the fitting curve of the scatter plot distribution. It shows that the number of passenger orders under different departure time fits the normal distribution.}
\label{fig_N2}
\end{figure}

To calculate the value of $P(T|X)$, we need to know the distribution of $(T|X)$. Fig. \ref{fig_N2} shows the scatterplot of the orders' departure time from a certain block within one day. We can find that the shape of the scatterplot is very similar to that of a Gaussian distribution. By studying many departure block and different days, we have similar observations. Therefore, we use Gaussian distribution to estimate the conditional probability of the departure time:
\begin{equation}
\label{equ8}
(T|X \textrm{=} {b_i}) \sim N({\mu _i},\sigma _i^2)
\end{equation}

Since variable $T$ takes circular values and then repeat, the mean $\mu_i$ and variance $\sigma_i$ cannot be estimated using traditional methods. We utilize the method proposed in \cite{Zhang2017taxi} to estimate the two parameters:
\begin{equation}
\left\{ \begin{array}{l}
\mathop {\min }\limits_\mu  \sum\limits_k^{N} {[|{{(|{t_k} \mathrm{-} \mu | \mathrm{-} \frac{N}{2}| \mathrm{-} \frac{N}{2}]}^2}} \\
s.t.\ \ \mu \mathrm{\in} [1,{N}]
\end{array} \right.
\end{equation}
\begin{equation}
{\sigma ^2} \textrm{=} \frac{1}{{{N_s} \mathrm{-} 1}}\sum\limits_{k = 1}^{N} {[|{{(|{t_k} \mathrm{-} \mu | \mathrm{-} \frac{N}{2}| \mathrm{-} \frac{N}{2}]}^2}}
\end{equation}

According to the normal distribution, the value of $P(T|X)$ can be calculated as:
\begin{equation}
P(T \textrm{=} t_k | X \textrm{=} b_i) \textrm{=} \int_{T \mathrm{\in} t_k} P(T|X \textrm{=} b_i)
\end{equation}

Furthermore, the value of $P(X)$ can be calculated as:
\begin{equation}
P(X{\rm{ = }}{b_i}){\rm{ = }}\frac{{freq({b_i})}}{{\sum\limits_i {freq({b_i})} }}
\end{equation}

Thus, the value of $P(X|T)$ can be calculated by combining the value of the $P(T|X)$ and the $P(X)$.

%--------------------------------------------------------------------------
\subsubsection{Calculation of Conditional Destination Probability}

In this subsection, we aim to model the probability distribution of each destination block. With Bayesian formula, the conditional departure probabilities $P(Y|X,T)$ can be expressed as:
\begin{equation}
\label{equ2}
P(Y \textrm{=} {b_j}|X \textrm{=} {b_i},T \textrm{=} {t_k}) \textrm{=} \frac{{P(X \textrm{=} {b_i},T \textrm{=} {t_k}|Y \textrm{=} {b_j}) \mathrm{\times} P(Y \textrm{=} {b_j})}}{{\sum_j {P(X \textrm{=} {b_i},T \textrm{=} {t_k}|Y \textrm{=} {b_j}) \mathrm{\times} P(Y \textrm{=} {b_j})} }}
\end{equation}
Two parameters need to be calculated in this formulation, $P(X,T|Y)$ and $P(Y)$. To calculate the probability of $P(X,T|Y)$, we need to know the distribution of $(X,T|Y)$. According to~\cite{Zhang2017taxi}, the same passenger tends to go to the same destination at the similar time and departure location. Thus, the distribution of the departure time and the departure longitude and latitude for each destination point of each passenger fit a three-dimensional Gaussian distribution, i.e., for a passenger $v$, $(Lat_v,Lng_v,T_v|Y_v \textrm{=} {b_j}) \sim {N_3}({\mu_v},{\Sigma_v})$, where $Lat$ is the departure latitude, $Lng$ is the departure longitude, $\mu$ and $\Sigma$ are the mean and covariance matrix of the three-variate Gauss distribution, respectively.

The distribution of the departure time and the departure location (longitude and latitude) for a destination $b_j$, can be regarded as the summary of the distributions of multiple passengers, whose destination locations are within $b_j$. Suppose the distribution of each passenger is independent, the sum of them is also consistent with Gaussian distribution \cite{eisenberg2008sum}. Hence, we can get the following hypothesis:
\begin{equation}
\label{equ5}
(Lat,Lng,T|Y\textrm{=} {b_j}) \sim {N_3}({\mu},{\Sigma})
\end{equation}
To verify this hypothesis, the Mardia's test \cite{mardia1970measures} can be employed to examine the goodness of fit.

With the distribution formulation,  we can calculate the $P(X,T|Y)$ with the definite integral:
\begin{equation}
\label{equ6}
P(X \textrm{=} {b_i},T \textrm{=} {t_k}|Y \textrm{=} {b_j}) \textrm{=} \int\limits_{\scriptstyle Lat,Lng \mathrm{\in} {b_i},\hfill\atop
\scriptstyle T\mathrm{\in}t_k} {P(Lat,Lng,T|Y \textrm{=} {b_j})}
\end{equation}

The probability of $P(Y)$ can be calculated as:
\begin{equation}
\label{equ7}
P(Y \textrm{=} {b_j}) \textrm{=} \frac{{freq({b_j})}}{{\sum_j {freq({b_j})} }}
\end{equation}

Thus, the value of $P(Y,X|T)$ can be calculated according to the above analysis.

\iffalse
\begin{table}
\caption{The key notations used in this paper}
\label{tab_notation}
\centering
\begin{tabular}{|c|c|c|}
\hline
Symbol & Meaning & Range \\
\hline
$X$, $Y$ & Departure, destination block & $[1, m]$  \\
\hline
$T$ & Departure time slot & $[1, n]$ \\
\hline
$U$ & Package set & \\
\hline
$V$ & Passenger set & \\
\hline
$W$ & Taxi set & \\
\hline
\end{tabular}
\end{table}
\fi

%--------------------------------------------------------------------------------------------------------------------------------------
%\subsection{Phase-\uppercase\expandafter{\romannumeral2}: Package Route Planning}
\subsection{Phase-ii: Package Route Planning}
\label{sec_pkg}

\begin{figure}
\centering
\includegraphics[width=2in]{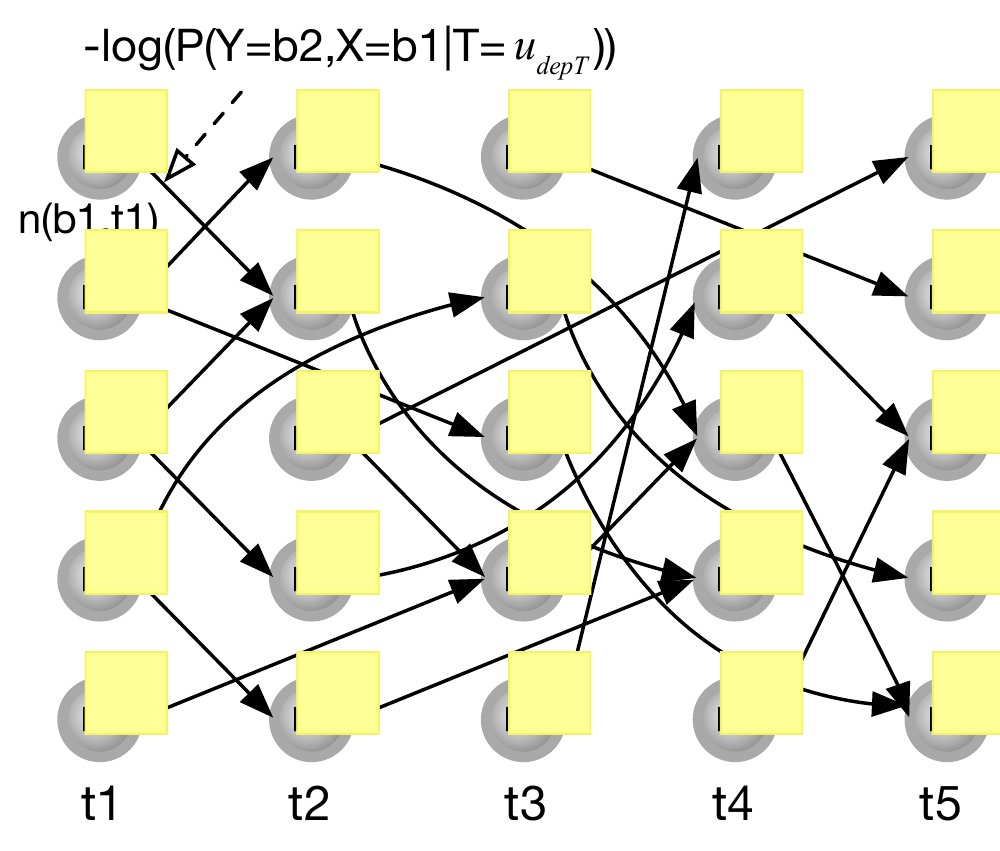}
\caption{An example of $G$. Each column denotes different blocks, each row denotes vary slots. The weight of the edge between nodes $n(b_i, t_k)$ and $n(b_j, t_g)$ is defined as $-log(P(Y,X|T))$, if $t_g\mathrm{-}t_k\textrm{=}\delta(b_i, b_j)$.}
\label{fig_newnet}
\end{figure}

In this subsection, we aim to select the route $r$ with the maximum probability $P(r)$ for a package request. As the NPD problem is NP-Complete, approximate route planning algorithm needs to be designed. The approximate algorithm should meet the following two kinds of requirements: computation efficiency and solution effectiveness. The route selection must be completed in real time because the passenger order dispatch must be timely in practice and the order dispatch results determine the package delivery path. The solution results should be effective enough, as too low package delivery rate makes our solution problematic. In order to design a route planning algorithm meeting the up two kinds of requirements, we propose three algorithms step by step. We first consider the prediction results and propose the Dijkstra-based optimal route planning algorithm (DOP) to find the ideal optimal route. Then we discuss the mismatching phenomenon between the prediction results and real-time order situations and propose a prediction-based sequential route planning algorithm (PSP). At last, we analyze the key ideas of the up two algorithms, and further design a heuristic sequential route planning algorithm (HSP), which takes appropriate tradeoff between the two kinds of requirements.

%--------------------------------------------------------------------------
\subsubsection{Dijkstra-based Optimal Planning}

When a package request appears, given the package information $u_{dep}$, $u_{depT}$, $u_{des}$, and order prediction results $P(Y,X|T)$, the platform needs to plan a route for the package. The objective is to maximize the probability of the planning route, with the constraint that the package should be delivered within a given time interval. In this subsection, we propose a Dijkstra-based algorithm to solve the formulated problem optimally in pseudo-polynomial time.

first, we construct a block-slot-based graph $G(\beta,\varepsilon)$ as illustrated in Fig. \ref{fig_newnet}, according to the prediction results $P$ and the demand network model $G_0(\beta_0, \varepsilon_0)$. In the graph $G$, each column represents the nodes in the same slot, each row represents the same node across different slots. An edge between two nodes $n(b_i, t_k) \mathrm{\in} \beta$ and $n(b_j, t_g) \mathrm{\in} \beta$ is constructed if $t_g \mathrm{-} t_k\textrm{=}\delta(b_i,b_j)$, and the edge weight is $-log(P(b_j, b_i | t_k))$.

second, the optimal route can be obtained through $G$. As each probability between nodes is less than $1$, there is no negative weight for all edges in $G$. The shortest path problem in graph $G$ with a given start node can be solved by the Dijkstra algorithm \cite{dijkstra1959note}. As we divide one day into multiple slots, with a limited $maxT$ in the constraint, the possible delivery destination node in the graph $G$ is numerable. For example, if $maxT\textrm{=}3$ slots, i.e., $u_{depT}\mathrm{<}u_{desT}\mathrm{\le}u_{depT}\mathrm{+}3$, the possible delivery destination node in the graph $G$ includes $n(u_{des},u_{depT}+1)$, $n(u_{des},u_{depT}+2)$, and $n(u_{des},u_{depT}+3)$. Therefore, the NPD problem can be solved by calculating the shortest path between the start node $n(u_{dep},u_{depT})$ and each destination node, and selecting the path with the minimum weight. 

\begin{theorem}
The Dijkstra-based optimal route planning algorithm solve problem NPD with the time complexity of $O(M^2\mathrm{\times}maxT^2)$.
\end{theorem}
\begin{proof}
The construction complexity of the adjacency matrix of $G$ is $O(M^2\mathrm{\times}maxT)$. The complexity of the Dijkstra algorithm is $O(M^2\mathrm{\times}maxT^2)$ \cite{cormen2009introduction}. Therefore, the complexity of the Dijkstra-based algorithm is $O(M^2\mathrm{\times}maxT^2)$. 
\end{proof}

The Dijkstra-based optimal route planning algorithm (DOP) has a pseudo-polynomial time complexity in the sense that the complexity is polynomial in the value of the problem input $maxT$, but is exponential in the length of the problem input, i.e., $log(maxT)$ \cite{garey1978strong,lin2018optimal}. Thus the solution is amenable for practical implementation. Note that Theorem 2 together with Theorem 1 shows that our problem NPD is actually \emph{NP-hard in the weak sense} \cite{lewis1983computers}.

\begin{algorithm}[!t]
\caption{Dijkstra-based Optimal Planning (DOP)}
\label{algo_dop}
\KwIn{order prediction results $P$, the graph $G_0$, $maxT$, the package request $u_{dep}, u_{depT}, u_{des}$}
\KwOut{the minimum route distance $mindist$, the optimal route $r_{opt}$}
{Construct $G$ based on the $P$ and $G_0$, with $maxT$;}\\
{$[distMat, routMat]=Dijkstra(G, G(u_{depT},u_{dep}))$;
//$[x,y]\textrm{=}Dijkstra(A,s)$ returns the distance and path of all destinations from start node $s$ in the graph $A$;}\\
{$mindist=min(distMat[u_{dep}+1:u_{depT}+maxT])$;}\\
{$r_{opt}=routMat(find(distMat==mindist),:)$;}
\end{algorithm}

However, although an optimal route can be selected based on the predicted probabilities, the selected route does not consider the real-time passenger orders. Since the optimal route is based on the statistical results, the prediction results may be \textbf{mismatched} under the real-time order situations. As it is hard to predict the future passenger situations accurately among the city-wide scope and across every time slot, the mismatching phenomenon between the prediction results and the real-time passenger orders will always exist. Moreover, even the passenger prediction of one hop is accurate enough, the cumulative mismatching error along the multiple hops of a route is inevitable (unless the prediction is $100\%$ accurate, which is almost impossible to realize without some strong assumptions). For example, suppose the average prediction accuracy of one hop passenger order is $90\%$, and there are four hops in a route, the prediction accuracy of this route is $(90\%)^4$, i.e., as low as $65.61\%$. Thus, a more realistic algorithm is needed to solve the NPD problem in online scenarios.

%--------------------------------------------------------------------------
\subsubsection{Prediction-based Sequential Planning}

In this subsection, we further propose a supplementary algorithm to tackle the mismatching problem, based on the DOP algorithm. The fundamental idea of the algorithm is that, when the mismatching occurs, we replan the ongoing route and select an optimal passenger order appeared to transport. Thus, we optimize the selected route step-by-step. On the one hand, if the online passenger orders correspond with the predicted results, and there is an expected order appears, the optimal route will be adopted. On the other hand, in each mismatched step where the predicted order of the optimal route does not appear in time, a proper passenger order should be selected. In this situation, a candidate set of potential passengers are first selected. If a passenger is within the same block and time slot with the package, the passenger is assumed as a potential order for the package. Within the candidate passengers, the probability from each passenger's destination to the package's destination is computed. The passenger with the maximum probability is regarded as the optimal one among the potential passengers and is selected to transport in this step.

\begin{algorithm}[!t]
\caption{Prediction-based Sequential Planning (PSP)}
\label{algo_psp}
\KwIn{online passengers $V$, order prediction results $P$, the graph $G_0$, $maxT$, the package request $u_{dep}, u_{depT}, u_{des}$}
\KwOut{package delivery route $r_{act}$}
$lenT=0$; $r_{act}=\emptyset$;\\
$r_{opt}=DOP(P, G_0, maxT, u_{dep}, u_{depT}, u_{des})$;\\
\While{$u_l \ne u_{des}$ and $lenT < maxT$}
{$V_{cand}\leftarrow \{\forall v| v_{dep}\textrm{=}u_l, v_{depT}\textrm{=}u_t\}$;\\
\eIf{ find($v \in V_{cand}\ \& \ v_{des}=u_{des}$)}
{deliver $v$ and $u$;}
{\eIf{find a passenger $\hat v$ ($\hat v \in V_{cand}$) along $r_{opt}$}
{$v\leftarrow \hat v$:}
{\For{$\hat v \in V_{cand}$}
{$(d, r)=DOP(P, G_0, maxT\mathrm{-}lenT\mathrm{-}(\hat v_{desT}\mathrm{-}\hat v_{depT}),$\protect\\
$ \hat v_{des}, \hat v_{desT}, u_{des})$;}
$v \leftarrow \hat v \ with\ min\ d$;\\
$r_{opt}=r$;}}
deliver $v$;\\
$lenT = lenT + (v_{desT} - v_{depT})$;\\
$r_{act}=[r_{act}; v]$;}
\end{algorithm}

The pseudo code of the prediction-based sequential planning algorithm (PSP) is presented in Algorithm \ref{algo_psp}. The $lenT$ is the current time length of the package delivery, $u_l$ and $u_t$ denote the current location and slot of the package, respectively. The $maxT$ denotes the maximum time length. If $lenT$ is less than $maxT$, the taxi will continue transporting passengers while carrying the package. If $lenT$ is larger than $maxT$ and the package has not reached the destination, the package delivery fails. For the failed packages, we request the taxi to deliver the packages directly to the destination without transporting any passengers.

The PSP algorithm tackles the mismatching problem of the DOP algorithm. However, the PSP algorithm also magnifies the defect of the DOP algorithm. As we analyzed before, the NPD problem is NP-Hard in the weak sense, and the DOP algorithm can find the optimal solution in Pseudo-polynomial time. Although the DOP algorithm is amenable for practical implementation, too many iterations in applying the DOP algorithm makes it \textbf{time consuming}. With the PSP algorithm, the DOP algorithm has to be repeated for $(N_{hops}\mathrm{\times}N_{orders})$ times. The $N_{hops}$ denotes the number of passenger orders transported during one package delivery; the $N_{orders}$ denotes the average order number within $V_{cand}$ in each hop. For example, if the computation time of the DOP algorithm for a package is $20$ seconds\footnote{The details of the time cost and average order number are introduced in Section \ref{sec_experiment}.}, the $N_{orders}$ is around $20$ and $N_{hops}$ is $4$, then the computation time of the PSP algorithm reaches $27$ minutes, which is too long to be used in practice. In realistic, the passenger dispatch (which is performed by the cloud-based platform) must be completed in real time, as the moving speed of taxis is fast and location update is frequent. Thus, a more time-efficient algorithm is needed to solve the NPD problem in real-life scenarios.

%--------------------------------------------------------------------------
\subsubsection{Heuristic Sequential Planning}

In Section \ref{sec_prob}, we formulated the NPD problem of the route probability maximization, which is the product of passenger probabilities along the route. According to our system model, the passenger probability in a slot is always less than $1$. As there are many blocks within our target area, the probability value is typically small. Thus the more hops of orders, the smaller the value of the probability product. This phenomenon determines that most of the selected route produced by the DOP algorithm only consists of one-hop passenger order, unless the one-hop passenger probability between the package origin and destination is too low. Therefore, if we only consider the one-hop probability when selecting the next-hop order from the candidate orders in each step, the time cost can be much reduced. Accordingly, we propose a heuristic sequential route planning algorithm (HSP) to speed up the calculation process.

\begin{algorithm}[!t]
\caption{Heuristic Sequential Planning (HSP)}
\label{algo_hsp}
\KwIn{online passengers $V$, order prediction results $P$, $maxT$, the package request $u_{dep}, u_{depT}, u_{des}$}
\KwOut{package delivery route $r^h_{act}$}
$lenT=0$; $r^h_{act}=\emptyset$;\\
\While{$u_l \ne u_{des}$ and $lenT < maxT$}
{$V_{cand}\leftarrow \{\forall v| v_{dep}\textrm{=}u_l, v_{depT}\textrm{=}u_t\}$;\\
\eIf{ find($v \in V_{cand}\ \& \ v_{des}=u_{des}$)}
{deliver $v$ and $u$;}
{$v\leftarrow min.\{\ -log(P(u_{des}, v_{des}|v_{desT}))\}$;}
deliver $v$;\\
$lenT = lenT + (v_{desT} - v_{depT})$;\\
$r^h_{act}=[r^h_{act}; v]$;}
\end{algorithm}

The fundamental idea of HSP are as follows. In each passenger selection step, if a passenger has the same destination as the package, the passenger will be delivered as the next hop. Otherwise, we greedily hope that the next passenger dispatch can make the taxi deliver the package successfully; thus, we select the passenger who has the minimum $-log(P(u_{des}, v_{des}|v_{desT}))$. The pseudo code of the HSP algorithm is presented in Algorithm \ref{algo_hsp}. In HSP, we replace the DOP algorithm with the $-log(P(u_{des}, v_{des}|v_{desT}))$, as multiple calling of the DOP algorithm incurs the high time complexity in the PSP algorithm. Note that, the HSP algorithm only perform well when the route selected by the DOP algorithm consists of one-hop passenger order. Thus the HSP algorithm may perform less satisfactorily than the PSP algorithm in a general view. Nevertheless, by adopting the HSP algorithm, the computation time can be reduced. Therefore, we conclude that if the computation resource is abundant, the PSP algorithm is the optimal adoption; otherwise, the HSP algorithm is a better choice.

%--------------------------------------------------------------------------------------------------------------------------------------
\section{Performance Evaluation}
\label{sec_experiment}

In this section, we evaluate the performance of \textbf{PPtaxi}. We first introduce the experimental setup, including the datasets, and the benchmark algorithms used for comparison. Then we show the results of algorithm effectiveness and efficiency. Note that, in the experiments of this section, the route planning algorithm of our \textbf{PPtaxi} framework utilizes the HSP algorithm unless otherwise specified.

\subsection{Experimental setup}

\iffalse
\begin{figure}[!t]
\centering
\includegraphics[width=2in]{figures//scope.pdf}
\caption{The target area in this paper.}
\label{fig_scope}
\end{figure}
\fi

\textbf{Datasets}: We use a real-world dataset\footnote{https://gaia.didichuxing.com} collected and published by DiDi Chuxing, which is a popular online taxi-taking platform in China. The dataset includes passenger order data and taxi trajectory data from $2016.11.1$ to $2016.11.30$ in the city of Chengdu, China. The information of each passenger order consists of the order ID, the departure time, the departure location, the destination, and the arrival time; the information of each taxi consists of the taxi ID, the order ID, the time stamps and locations with a sampling rate of $2\mathrm{\sim}4$ seconds. 

We target at the area with the densest passenger orders, the longitude from $104$ to $104.12$ and the latitude from $30.6$ to $30.72$, which is the most central area in Chengdu. Then, the target area is divided into $10 \mathrm{\times} 10$ blocks, the size of each block is around $1.56\ km^2$, and one day is split into $144$ slots, the length of a time slot is set as $10$ minutes. As the distribution of passenger orders shows variable regularities in different days \cite{lin2018optimal,Zhang2017taxi}, we choose the data in the first four Tuesdays as the training set to calculate the order probabilities and the data in the fifth Tuesday as the testing set to evaluate our framework. Within the target area, there are total $540,864$ orders in the training set, $150,412$ orders and $40,772$ taxis in the testing set. Moreover, due to the specialty of the DiDi platform, where most of the taxis are private cars, the average working time length of each driver is limited, around $25.15$ slots in our statistics. The statistics of the datasets and parameters are listed in the Tab. \ref{tab_dataset}, where the $lng.$ denotes the longitude and the $lat.$ denotes the latitude. 

\iffalse
%\makeatletter\def\@captype{figure}\makeatother
\begin{figure}
\begin{minipage}{.4\textwidth}
\centering
\includegraphics[width=2in]{figures//scope.pdf}
\caption{The target area in this paper.}
\label{fig_scope}
\end{minipage}
%\makeatletter\def\@captype{table}\makeatother
\begin{minipage}{.5\textwidth}
\centering
\begin{tabular}{|c|c|c|}
\hline
 &  Properties & Statistics  \\
 \hline
\multirow{3}*{Parameters} & target area & \tabincell{c}{$lng.(104,104,12)$\\$lat.(30.6,30.72)$} \\
%\multirow{3}*{Parameters} & target area & $lng.(104,104,12), lat.(30.6,30.72)$ \\
\cline{2-3}
~ & block size & $1.56\ km^2$ \\
\cline{2-3}
~ & slot length & $10\ mins$ \\
\hline
Training set & order number & $540,864$ \\
\hline
\multirow{3}*{Testing set} & order number & $150,412$ \\
\cline{2-3}
~ & taxi number & $40,772$ \\
\cline{2-3}
~ & \tabincell{c}{average working time\\of each taxi} & $25.15$ slots \\
\hline
\end{tabular}
%\caption{Statistics of the Datasets and Parameters}
\captionof{table}{Statistics of the Datasets and Parameters}
\label{tab_dataset}
\end{minipage}
\end{figure}
\fi

\begin{table}[!t]
%\small
\caption{Statistics of the Datasets and Parameters}
\label{tab_dataset}
\centering
\begin{tabular}{|c|c|c|}
\hline
 &  Properties & Statistics  \\
 \hline
\multirow{3}*{Parameters} & target area & \tabincell{c}{$lng.(104,104,12)$\\$lat.(30.6,30.72)$} \\
\cline{2-3}
~ & block size & $1.56\ km^2$ \\
\cline{2-3}
~ & slot length & $10\ mins$ \\
\hline
Training set & order number & $540,864$ \\
\hline
\multirow{3}*{Testing set} & order number & $150,412$ \\
\cline{2-3}
~ & taxi number & $40,772$ \\
\cline{2-3}
~ & \tabincell{c}{average working time\\of each taxi} & $25.15$ slots \\
\hline
\end{tabular}
\end{table}

\textbf{Package Requests:} Since the datasets do not contain information about package delivery, we generate the package information manually. We assume that the package number is related with populations in an area, and more inhabitants often denote more packages. Thus we generate the package information based on population densities in different districts. The target area in our experiment covers five districts in Chengdu, i.e., Jinniu district, Qingyang district, Wuhou district, Chenghua district, and Jinjiang district. The population density ($2015$) in each district, which can be found on the Chengdu Bureau of Statistics Internet\footnote{http://www.cdstats.chengdu.gov.cn}, is $6990$, $9822$, $8755$, $6623$, and $8233$, respectively. Accordingly, there are around $600,000$ inhabitants in our target area.

\begin{figure*}[t]
\centering
\subfigure[$maxT\mathrm{=}3\ hours$]{
\label{fig_srt3h}
\includegraphics[width=0.4\textwidth]{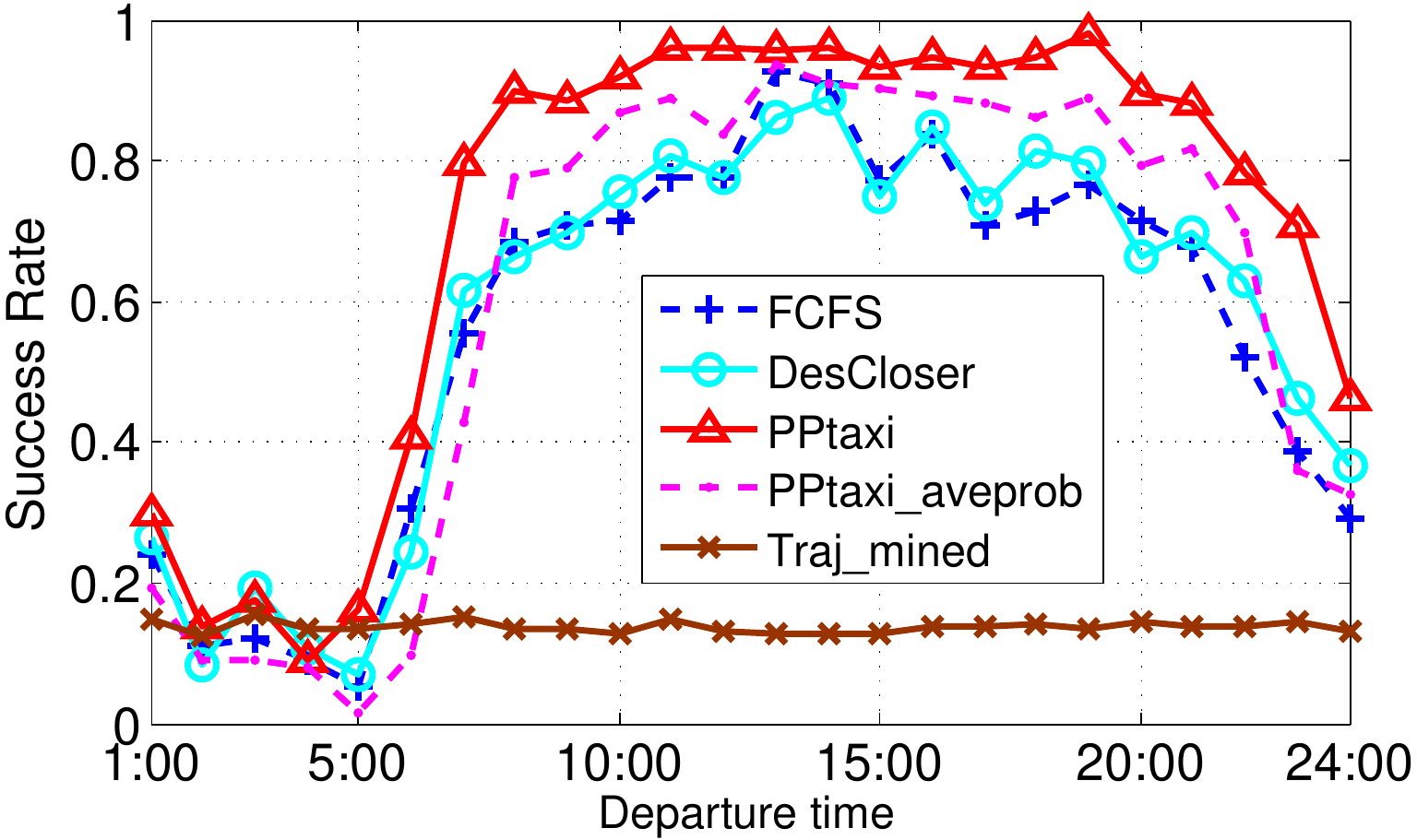}}
\subfigure[$maxT\mathrm{=}10\ hours$]{
\label{fig_srt10h}
\includegraphics[width=0.4\textwidth]{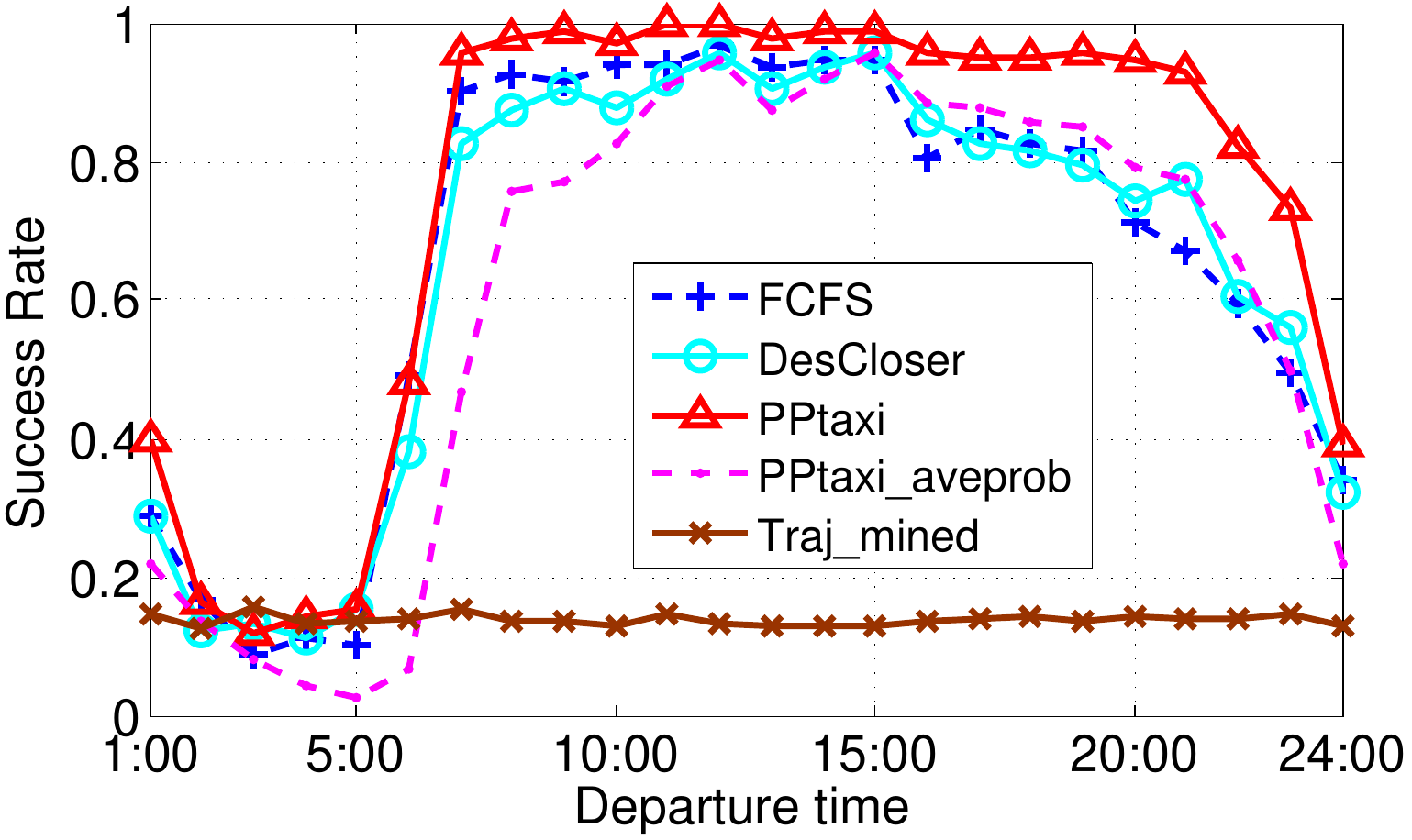}}
\caption{The Success Rate over departure time}
\label{fig_srt}
\end{figure*}

\textbf{Benchmarks}: To show the performance of our proposed framework, we compare it with the following benchmark algorithms. Note that, all the benchmarks are within the non-stop strategy, i.e., a taxi should continue to transport passengers until the package is delivered.

\begin{itemize}
\item \textbf{Traj\_mined}: We adopt the generated package requests on the trajectory dataset, and evaluate the package-transportation ability in today's taxi-booking platform. If the departure location and destination of a package request can match to one same driver's ID within the trajectory dataset, we think the package can be delivered successfully.
\item \textbf{FCFS}: This algorithm adopts the First Come First solution strategy. The taxi, which is carrying a package, always choose the first appeared passenger order in each step. This algorithm is actually a random walk strategy.
\item \textbf{DesCloser}: In each step of this algorithm, the taxi always select the passenger order who heads to somewhere closest to the package destination. This algorithm is actually a distance-based greedy strategy.
\item \textbf{PPtaxi\_aveprob}: In this algorithm, the taxi route planning algorithm is the same as \textbf{PPtaxi}, as the prediction-based sequential planning algorithm. Nevertheless, the passenger  order probability here is simply computed by the average value of historical data, i.e.,
\[P(Y \textrm{=} {b_j},X \textrm{=} {b_i}|T \textrm{=} {t_k}) \textrm{=} \frac{{freq(Y \textrm{=} {b_j},X \textrm{=} {b_i}|T \textrm{=} {t_k})}}{{\sum\limits_{i,j} {freq(Y \textrm{=} {b_j},X \textrm{=} {b_i}|T \textrm{=} {t_k})} }}\]
\end{itemize}

\textbf{Metric}: We adopt the \emph{Success Rate} (SR) as the metric to evaluate our algorithm and the benchmarks. The success rate is defined as the ratio of the number of packages successfully delivered within a given deadline to the total number of packages, i.e., 
\[Success\ Rate = \frac{{|u|u_{desT}\mathrm{-}u_{depT} \mathrm{\le} maxT|}}{{|U|}}\]
where the $u_{desT}$ and $u_{depT}$ denote the delivery time and departure time of package $u$, respectively; the $maxT$ is the threshold on the delivery time, the $U$ denotes the number of all packages.

\textbf{Environment}: All the evaluations in this paper are run in Matlab R2014a on an Intel Core i5-7400 PC with 16-GB RAM and Windows 7 operation system. 

\subsection{Performance Evaluation}

The following questions are of our interests:
\begin{itemize}
\item How well does the PPtaxi perform overall?
\item How well does the PPtaxi perform under different time constraints? 
\item How well does the PPtaxi perform with package number increase?
\item How many computational resources are required to generate the response for a package delivery request?
\item How reasonable is our problem formulation?
\end{itemize}

\textbf{(1) How well does the PPtaxi perform overall?}

In this experiment, we choose $100$ $(u_{dep}, u_{des})$ pairs and $u_{depT}$ at each hour i.e., in total $100\mathrm{\times}24\textrm{=}2,4000$ packages. The package departure locations are generated according to the density ratio in different blocks; the package destinations are generated randomly. We evaluate our algorithm and benchmarks on those $2,4000$ instances and compute the average performance. Fig. \ref{fig_srt} shows the Success Rate under the different departure time of the packages.

First, we can find that the SR of \textbf{PPtaxi} can be at highest $100\%$ when the $maxT$ is $10$ hours, which denotes that all of the packages can be delivered successfully through our non-stop package delivery solution. Compare with the benchmarks, the SR of \textbf{PPtaxi} is always higher than the three benchmark algorithms no matter what the departure time is, which proves the efficiency of our framework. The SR advantage of \textbf{PPtaxi} can reach $46.9\%$ at most. The SR advantage of the \textbf{PPtaxi} is higher than that of \textbf{PPtaxi\_aveprob} proves the efficiency of our passenger prediction method. The SR results of \textbf{FCFS} fluctuate obviously, while our non-stop framework is more stable with different departure time. Moreover, the performances of \textbf{DesCloser} and \textbf{FCFS} are very similar. The reason is that we always select the passenger order heading to somewhere closest to the package destination among all the orders in each step, instead of selecting a closer (than the current location) one, since maybe no orders meet this requirement. 

Second, the SR of \textbf{Traj\_mined} (which is around $19\%$) is far below that of \textbf{PPtaxi}. This phenomenon is determined by the feature of the dataset. Due to the specialty of the DiDi platform, where most of the taxis are private cars, the average working-time length of each driver is limited, around $25.15$ slots according to our statistics in Tab. \ref{tab_dataset}. It is far less than that of professional taxis. Moreover, today's passenger assignment is aiming to maximize the income of the platform, without considering the package delivery. Therefore, the success rate mined through the trajectory dataset is very low. It states that the passenger assignment in our framework can improve the package delivery apparently.

Third, the SR of all four algorithms is lower than $50\%$ when the departure time is within $0\mathrm{:}00$ to $6\mathrm{:}00$, and then increases apparently and becomes higher than $60\%$ while the departure time is within $7\mathrm{:}00$ to $23\mathrm{:}00$. The reason can be seen in the Fig. \ref{fig_aven}, which illustrates the average order number in each slot (among all the blocks). We can see that the former period is the non-rush hours, the passenger number is insufficient; the latter period is the rush hours, and there are enough passenger orders to provide the ridesharing chances to deliver a package. When the $maxT$ is $3$ hours, the SR with our solution can reach $95\%$ on average during the daytime. 

\textbf{(2) How well does the PPtaxi perform under different time constraints? }

Fig. \ref{fig_srmaxt} illustrates the SR of four algorithms under different settings of the $maxT$, the threshold of delivery time. In this experiment, we choose $100$ ($u_{dep}, u_{des}$) pairs of packages under each departure time setting. In Fig. \ref{fig_srmaxt8c}, the departure time of each package (denoted as $u_{depT}$) is $8\mathrm{:}00$; in Fig. \ref{fig_srmaxt15c}, the departure time is set as $15\mathrm{:}00$. 

\begin{figure}
\centering
\subfigure[$u_{depT}\mathrm{=}8\mathrm{:}00$]{
\label{fig_srmaxt8c}
\includegraphics[width=0.32\textwidth]{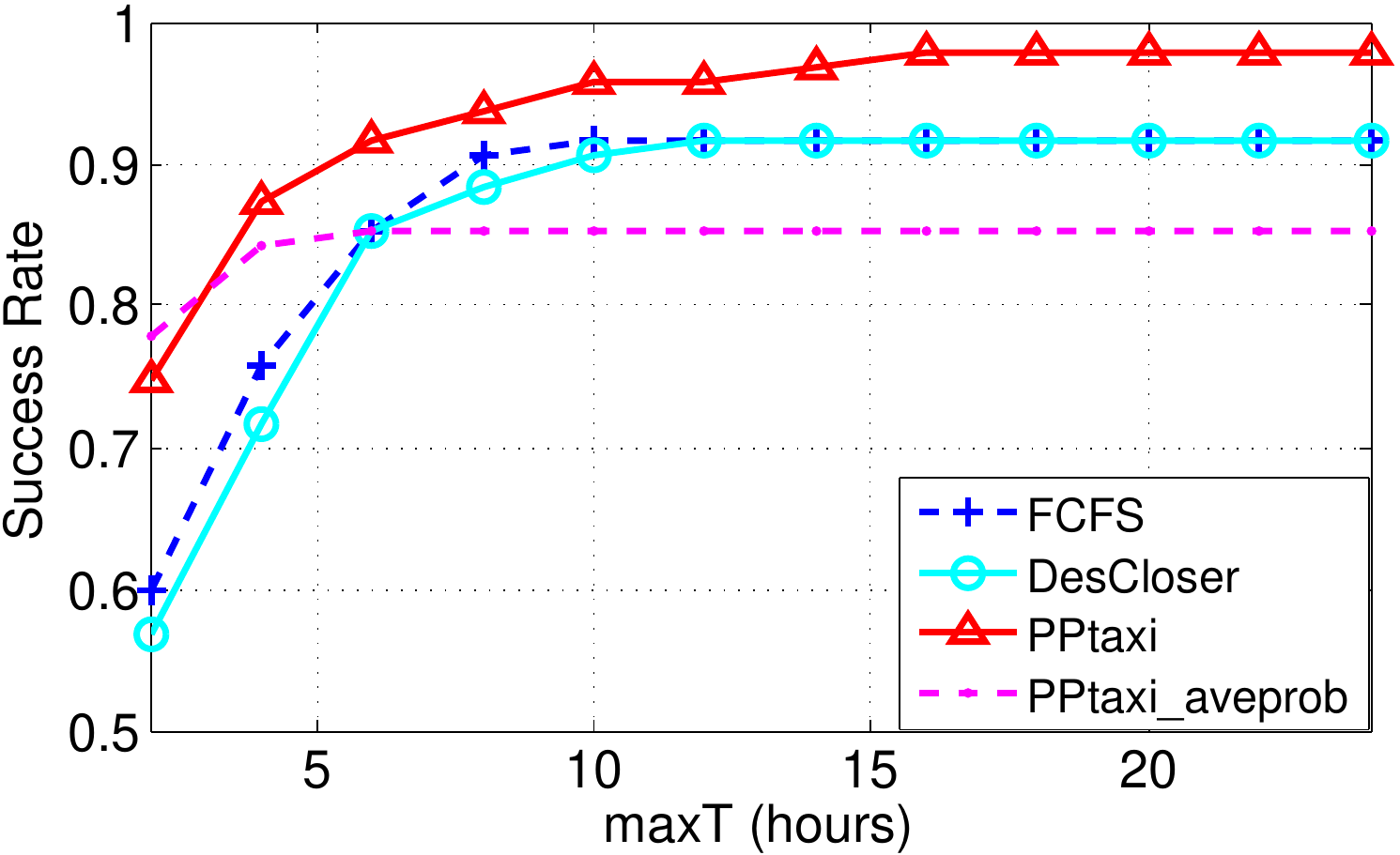}}
\subfigure[$u_{depT}\mathrm{=}15\mathrm{:}00$]{
\label{fig_srmaxt15c}
\includegraphics[width=0.32\textwidth]{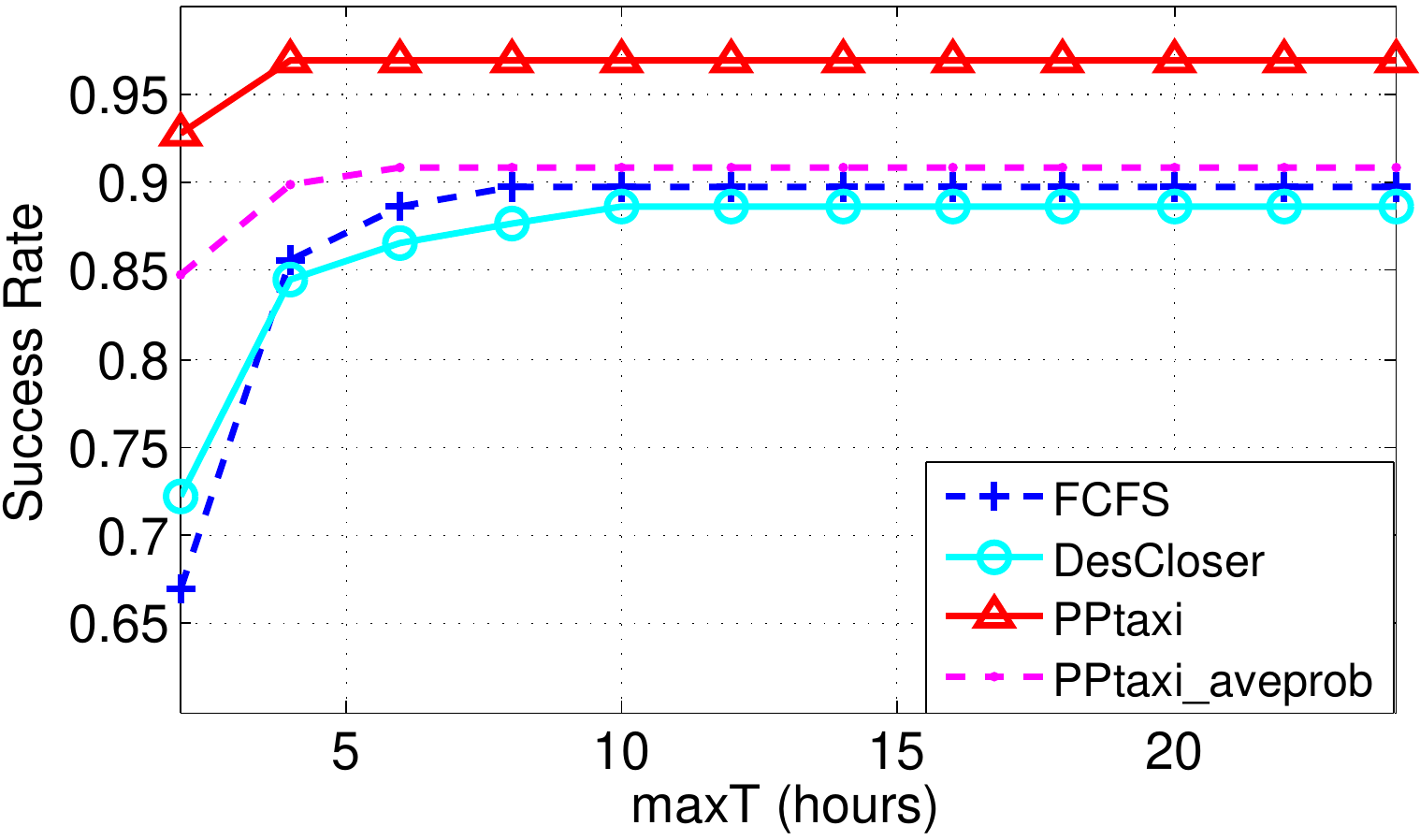}}
\caption{The Success Rate over time constraints}
\label{fig_srmaxt}
\end{figure}

First, we can see that the SR of \textbf{PPtaxi} increases with the growth of $maxT$. The reason is that longer $maxT$ means the more potential delivery routes. The SR increase various under different departure time, as the passenger distributions are different from each other and our prediction shows a various level of accuracy. When the departure time is $15\mathrm{:}00$, the SR increase of \textbf{PPtaxi} is limited as the $maxT$ increases. The reason is that, after $5$ hours, the order number shows a noticeable decline, then after $8$ hours, the city enters the non-rush hours. Thus, the $maxT$ setting larger than $8$ hours brings little chance. Specifically, the highest SR of \textbf{PPtaxi} is always lower than $1$. The reason is that we generate the package destinations randomly without considering the area functions. Thus destinations without residents may incur impossible delivery since no orders appear heading to these destinations.

Second, compared with the other three algorithms, the SR of \textbf{PPtaxi} is always the highest, no matter what the $maxT$ is. The results of \textbf{PPtaxi\_aveprob} algorithm show unstable performance. Sometimes it is higher than the SR of \textbf{FCFS}, while at other times it is lower than the latter one. That phenomenon denotes that the prediction through historical average values is not accurate enough, as the cumulative error of that algorithm induces the inaccuracy of multi-hop-passenger prediction. Furthermore, we are surprised to find that the \textbf{FCFS} algorithm performs very well (the SR is $90\%$) when the $maxT$ is long enough, i.e., more than $8$ hours. However, we have to emphasize that our non-stop framework request the taxi to keep on transporting passengers until the package is delivered successfully, so it is more suitable for the short-term delivery. In short-term route planning, the advantage of \textbf{PPtaxi} is more evident than the other algorithms.

\textbf{(3) How well does the PPtaxi perform with package number increase?}

\begin{figure}
\centering
\subfigure[$maxT\textrm{=}3\ hours$]{
\label{fig_srp3h}
\includegraphics[width=0.32\textwidth]{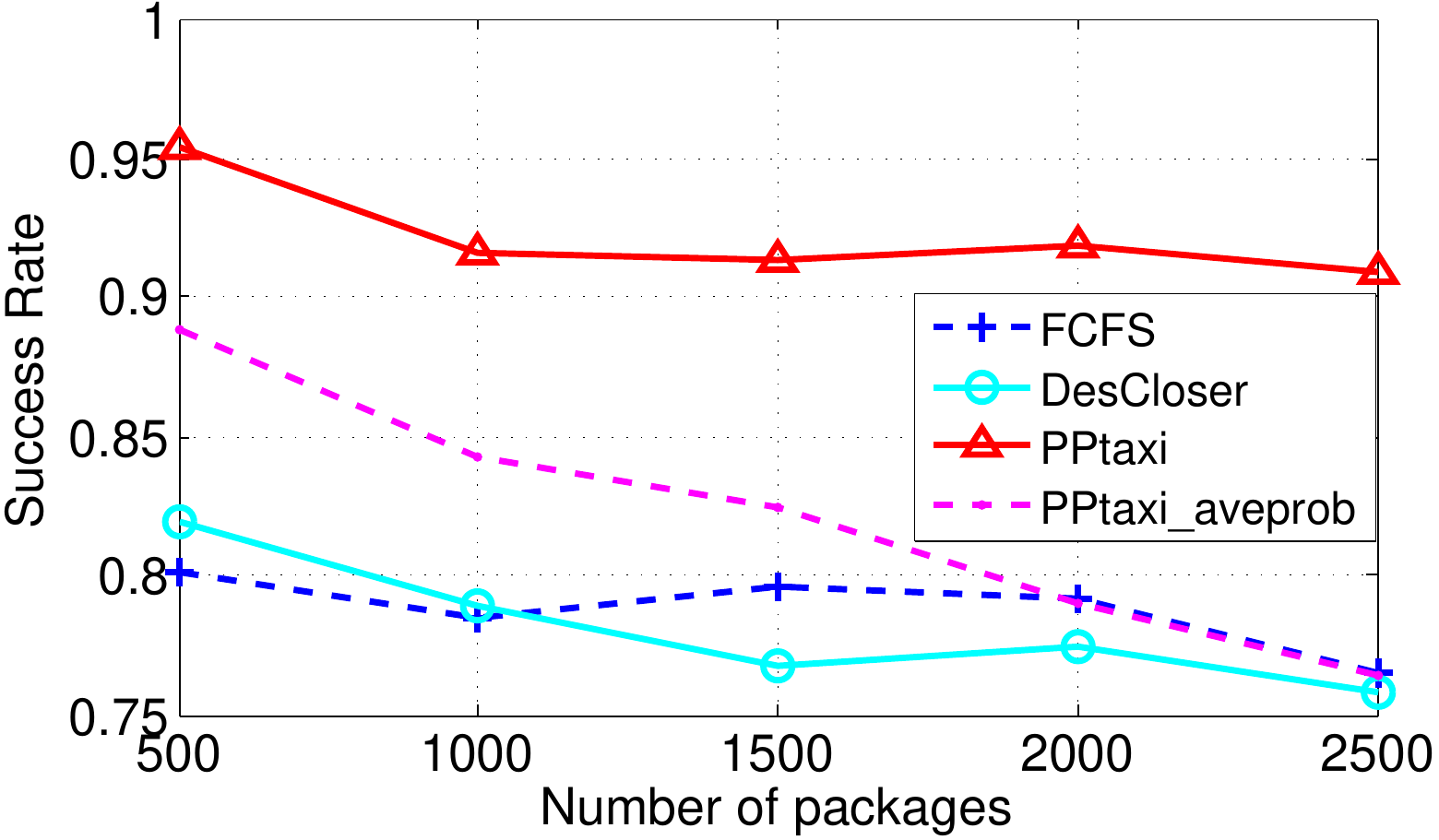}}
\subfigure[$maxT\textrm{=}10\ hours$]{
\label{fig_srp10h}
\includegraphics[width=0.32\textwidth]{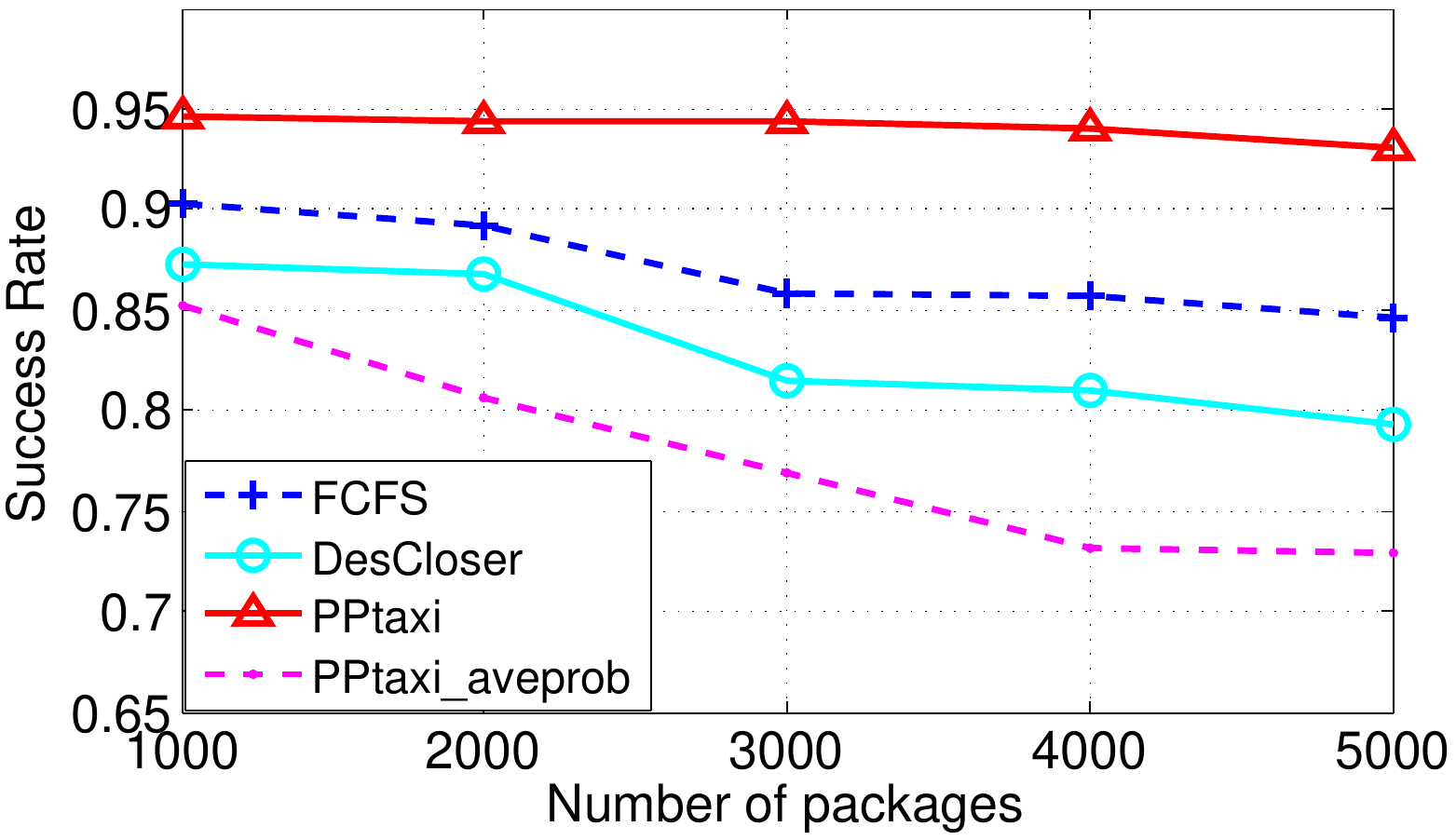}}
\caption{The Success Rate over package number. ($u_{depT}\textrm{=}15\mathrm{:}00$)}
\label{fig_srp}
\end{figure}

Fig. \ref{fig_srp} shows the performance of the four algorithms under different package numbers. In this experiment, the departure time of all packages is set as $15\mathrm{:}00$. In Fig. \ref{fig_srp3h} and Fig. \ref{fig_srp10h}, we set the delivery time constraint as $3$ hours and $10$ hours, respectively. 

First, we can see that, when the time constraint is $3$ hours, the SR of \textbf{PPtaxi} still can reach $90\%$ with $2500$ packages generated within one hour, while the SR of the other three algorithms is lower than $80\%$. As discussed in Fig. \ref{fig_srmaxt}, with the $maxT$ increases, the SR rises, and more packages can be delivered successfully. When the time constraint is $10$ hours, $93\%$ of the $5000$ packages can be delivered successfully through our \textbf{PPtaxi} framework. Meanwhile, the SR of the other three algorithms is at most $85\%$.

Second, the SR of all four algorithms decreases with the growth of the package number. The reason is that multiple package delivery simultaneously may reduce resource congestion, since different delivery routes may need the same passenger order. In this paper, we propose to optimize the route planning of each package one by one instead of giving a globally optimal solution for all packages. Thus the resource congestion phenomenon is not in-depth discussed; this part will be the primary focus of our future work. Nevertheless, the experiment results show that the decrease of our \textbf{PPtaxi} with the package number growth is far less than those of the other three algorithms, which indicate the superiority of our solution.

\textbf{(4) How many computational resource is required to generate the response for a package delivery request?}

\begin{table}[htb]
%\small
\caption{The cost and performance comparison of the proposed three route planning algorithms}
\label{tab_timecost}
\centering
\begin{tabular}{|c|c|c|c|}
\hline
 &  OPT & PSP & HSP  \\
 \hline
 average time cost & 22.5s & 952.8s& 0.007s \\
 \hline
 Success Rate & 34\% & 94\% & 92\% \\
\hline
\end{tabular}
\end{table}

In this subsection, we evaluate the time cost of our proposed three algorithms, including OPT, PSP, and HSP. Since the OPT algorithm does not consider the mismatching problem, the SR result is low in the real-life scenarios. The PSP algorithm calls the OPT algorithm many times while the latter one is a pseudo-polynomial-time algorithm, thus the PSP algorithm is not efficient enough to be used in the real world. The HSP algorithm is a heuristic algorithm that tradeoff between the computation efficiency and solution effectiveness. In this experiment, we set the departure time as $8\mathrm{:}00$, the delivery time constraint as $3$ hours, and evaluate the average performance of $50$ packages. The results about the time cost and performance comparison of the proposed three rout planning algorithms can be seen in Table \ref{tab_timecost}. 

We can see that the SR of the OPT algorithm is $34\%$, which is very low due to the prediction inaccurate. The average time cost of one package delivery request by utilizing the OPT algorithm is around $22.5s$. The SR of the PSP algorithm is the highest ($94\%$) since the mismatching problem is tackled and the prediction based optimal route is adopted. That result also indicates that even though the passenger prediction is not accurate enough, our route planning algorithm can make the success rate of package delivery reaches a satisfactory value. However, the average time cost of one package delivery request by utilizing the PSP algorithm is as high as $15.88$ minutes, which is intolerable in the real world due to the rapid dynamic of passenger distributions. The SR of the HSP algorithm is less than that of the PSP algorithm, but it still can reach $92\%$. Moreover, the average time cost of one package delivery request by utilizing the HSP algorithm is only $0.007s$. Therefore, the HSP algorithm improves the computation efficiency significantly by sacrificing a small portion of the performance, which makes itself amenable for practical implementation.

\textbf{(5) How reasonable is our problem formulation?}

We formulate the non-stop problem of the route probability maximization and translate the problem into the shortest route problem. Fig. \ref{fig_ap} shows the average $-log(P(r))$ (denoted as AP) of the delivery routes resulted by different algorithms under different parameter settings. $P(r)$ denote the probability of the route $r$, i.e., the probability product of the passenger orders along the route. As the probability calculation method of \textbf{PPtaxi} and \textbf{PPtaxi\_aveprob} are different, we did not list the result of the latter algorithm. 

Fig. \ref{fig_apt} illustrates the AP result of the three algorithms under different departure time; the time constraint is set as $10$ hours, which corresponds to the results in Fig. \ref{fig_srt10h}. We can see that when the departure time is within $7\mathrm{:}00$ to $23\mathrm{:}00$, the AP of \textbf{PPtaxi} is always lower than that of \textbf{FCFS} or \textbf{DesCloser}. More specifically, combined with the results of Fig. \ref{fig_srt}, within $7\mathrm{:}00$ to $20\mathrm{:}00$, the AP of the three algorithms decreases and the SR increases; within $20\mathrm{:}00$ to $23\mathrm{:}00$, the AP increases and the SR decreases. This phenomenon can prove that our formulation of maximizing the probability cumulation of each hop within a route is reasonable. However, when the departure time is within $0\mathrm{:}00$ and $6\mathrm{:}00$, the AP results of the three algorithms are almost the same, and it did not show the opposite tendency between SR and AP. The reason is that the passenger number in this period is limited and the regularity is elusive, which leads to the inaccuracy of the passenger prediction. As most packages are sent in the daytime in the real-world, our \textbf{PPtaxi} is amenable for practical implementation. Furthermore, we believe that by combing more data sources, the prediction can be more accurate and the performance of \textbf{PPtaxi} can be further improved.

Similarly, Fig. \ref{fig_apmaxt} and Fig. \ref{fig_app} show the AP result of the three algorithms under different $maxT$ and package numbers, respectively. We can see that the AP of \textbf{PPtaxi} is always lower than that of \textbf{FCFS} or \textbf{DesCloser}, which meets our expectations. The SR of each algorithm increases with the growth of $maxT$ since more hops are needed to deliver more packages successfully. The SR of \textbf{FCFS} and \textbf{DesCloser} increases dramatically, because there is no route optimization in those two algorithms.

\textbf{(6) The demonstration of the order dataset.}

\begin{figure*}[!t]
\centering
\subfigure[$maxT\textrm{=}10\ hours$]{
\label{fig_apt}
\includegraphics[width=0.32\textwidth,height=1.4in]{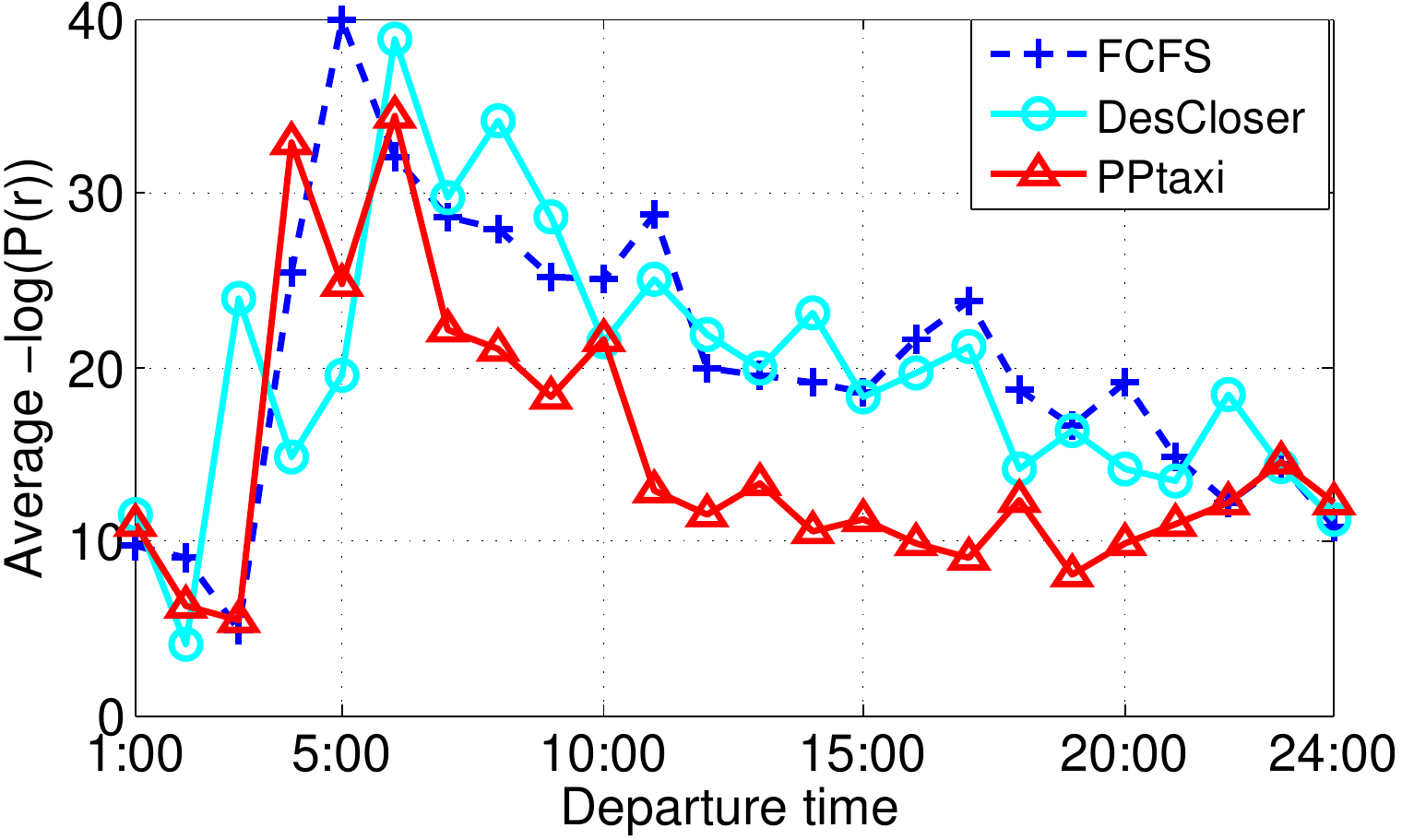}}
\subfigure[$u_{depT}\mathrm{=}15\mathrm{:}00$]{
\label{fig_apmaxt}
\includegraphics[width=0.32\textwidth,height=1.4in]{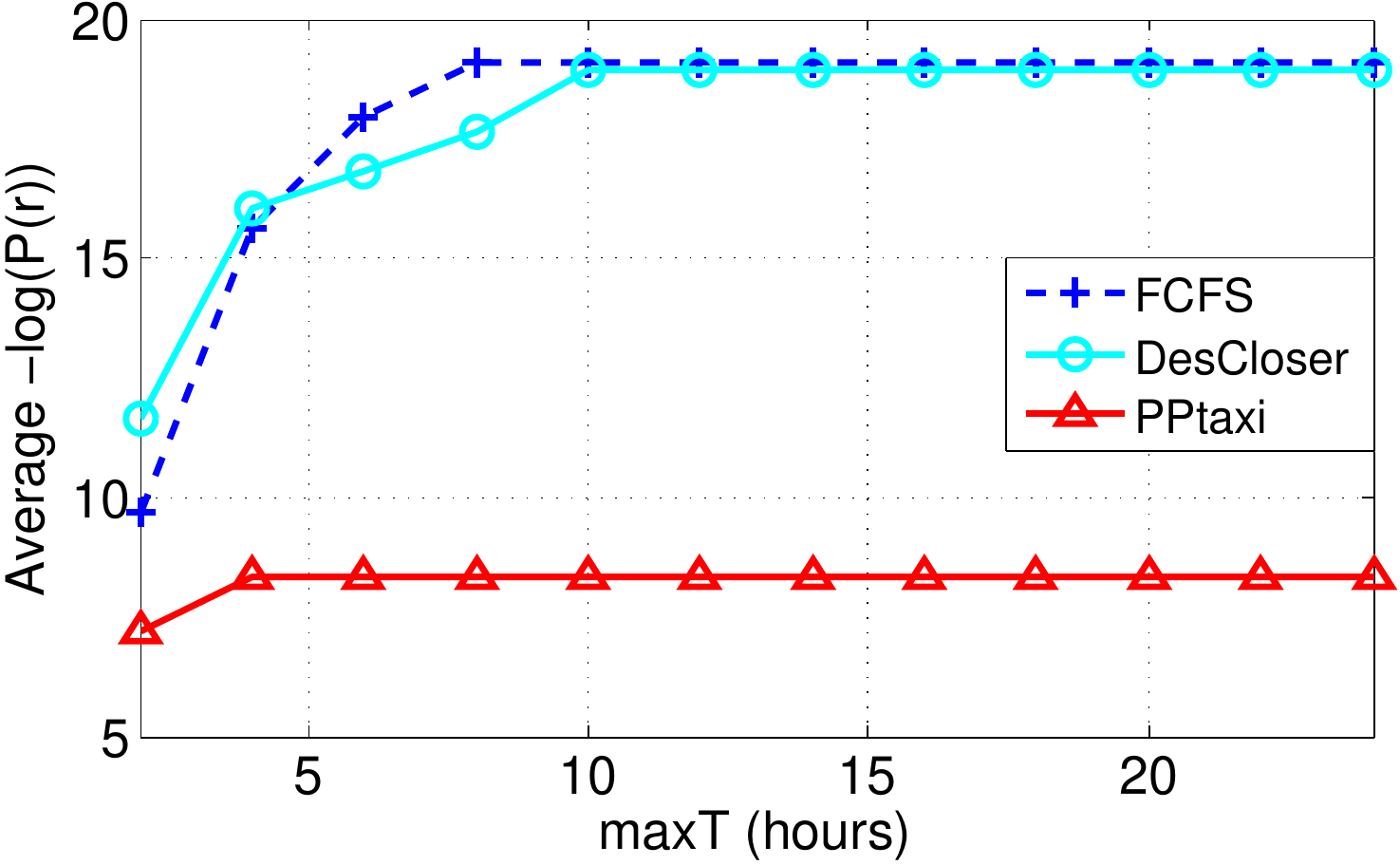}}
\subfigure[$maxT\textrm{=}3\ hours, u_{depT}\mathrm{=}15\mathrm{:}00$]{
\label{fig_app}
\includegraphics[width=0.32\textwidth,height=1.4in]{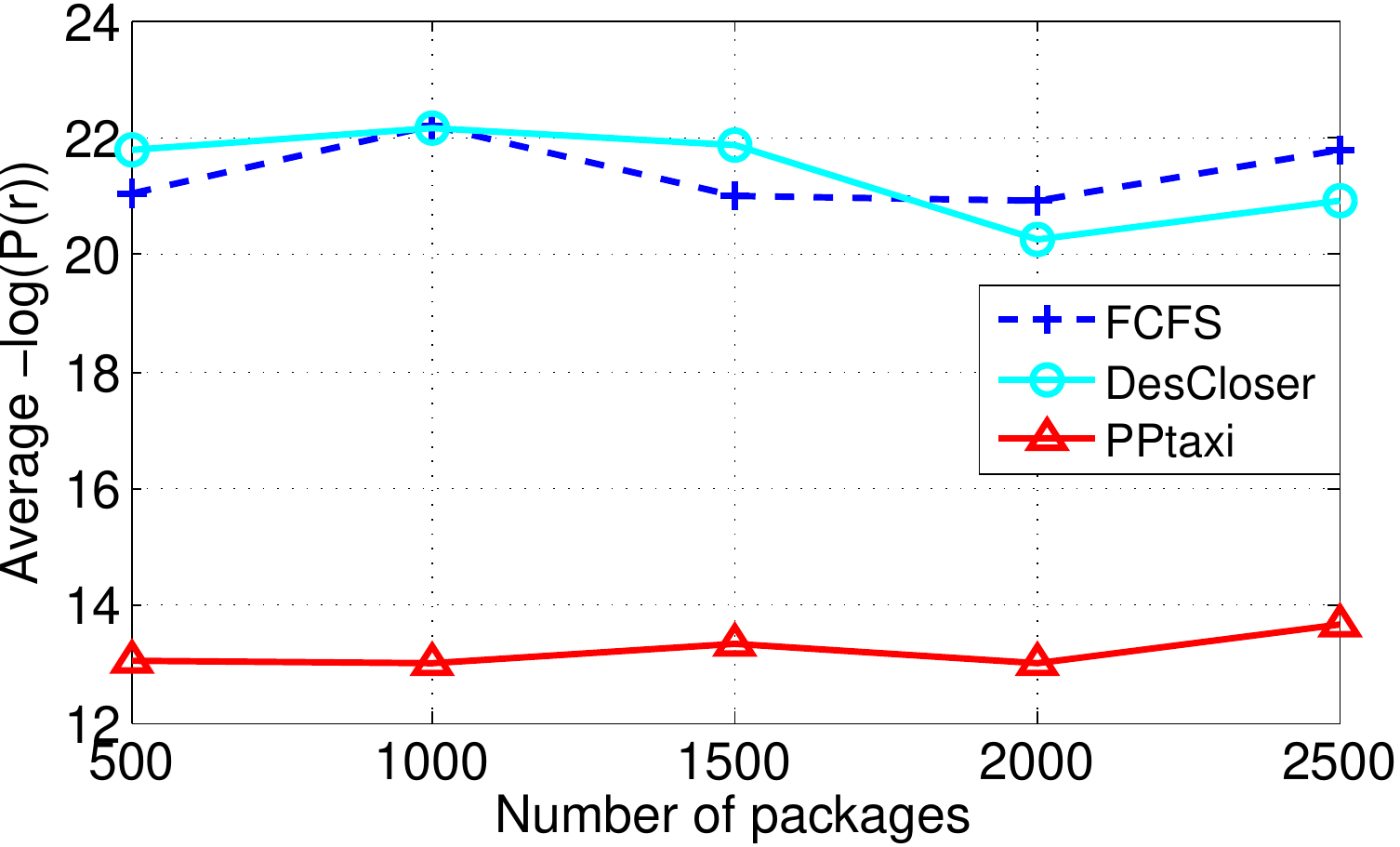}}
\caption{The Average $-log(P(r))$ results over different parameter settings.}
\label{fig_ap}
\end{figure*}
\begin{figure*}
\begin{minipage}[t]{0.33\textwidth}
\centering
\includegraphics[width=\textwidth,height=1.4in]{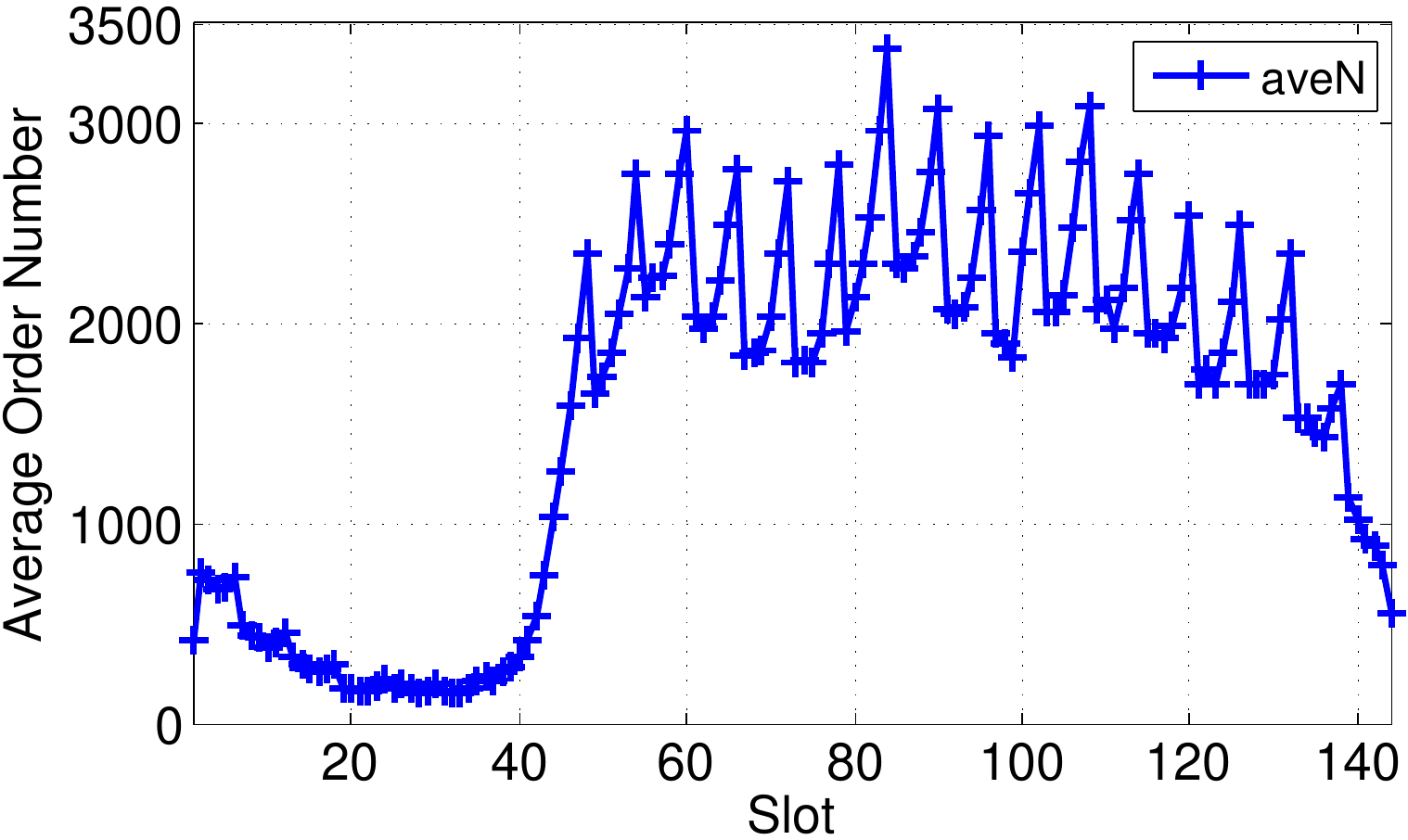}
\caption{The average order number in each slot.}
\label{fig_aven}
\end{minipage}
\begin{minipage}[t]{0.33\textwidth}
\centering
\includegraphics[width=\textwidth,height=1.4in]{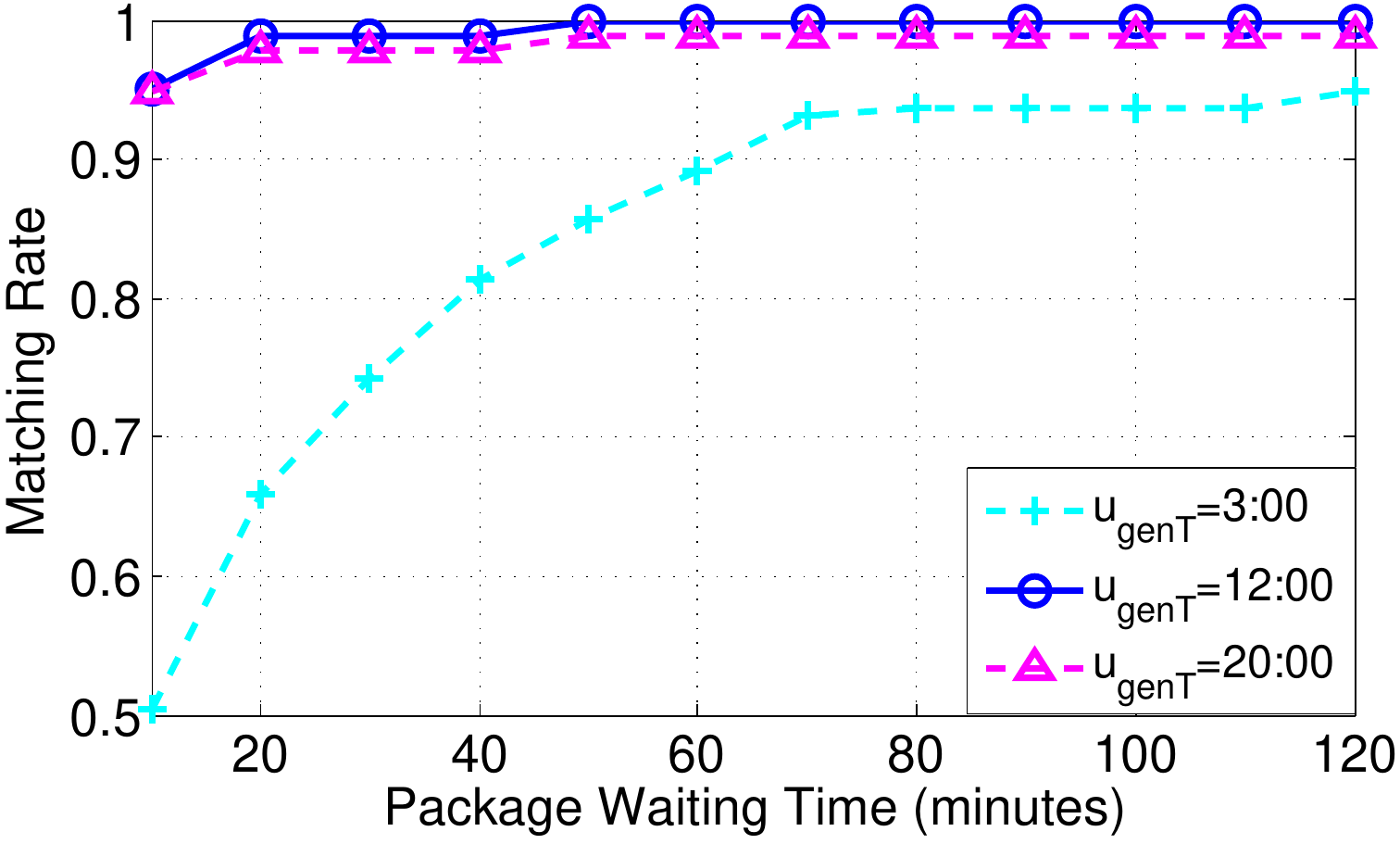}
\caption{The package-taxi matching.}
\label{fig_ptnum}
\end{minipage}
\begin{minipage}[t]{0.33\textwidth}
\centering
\includegraphics[width=\textwidth,height=1.4in]{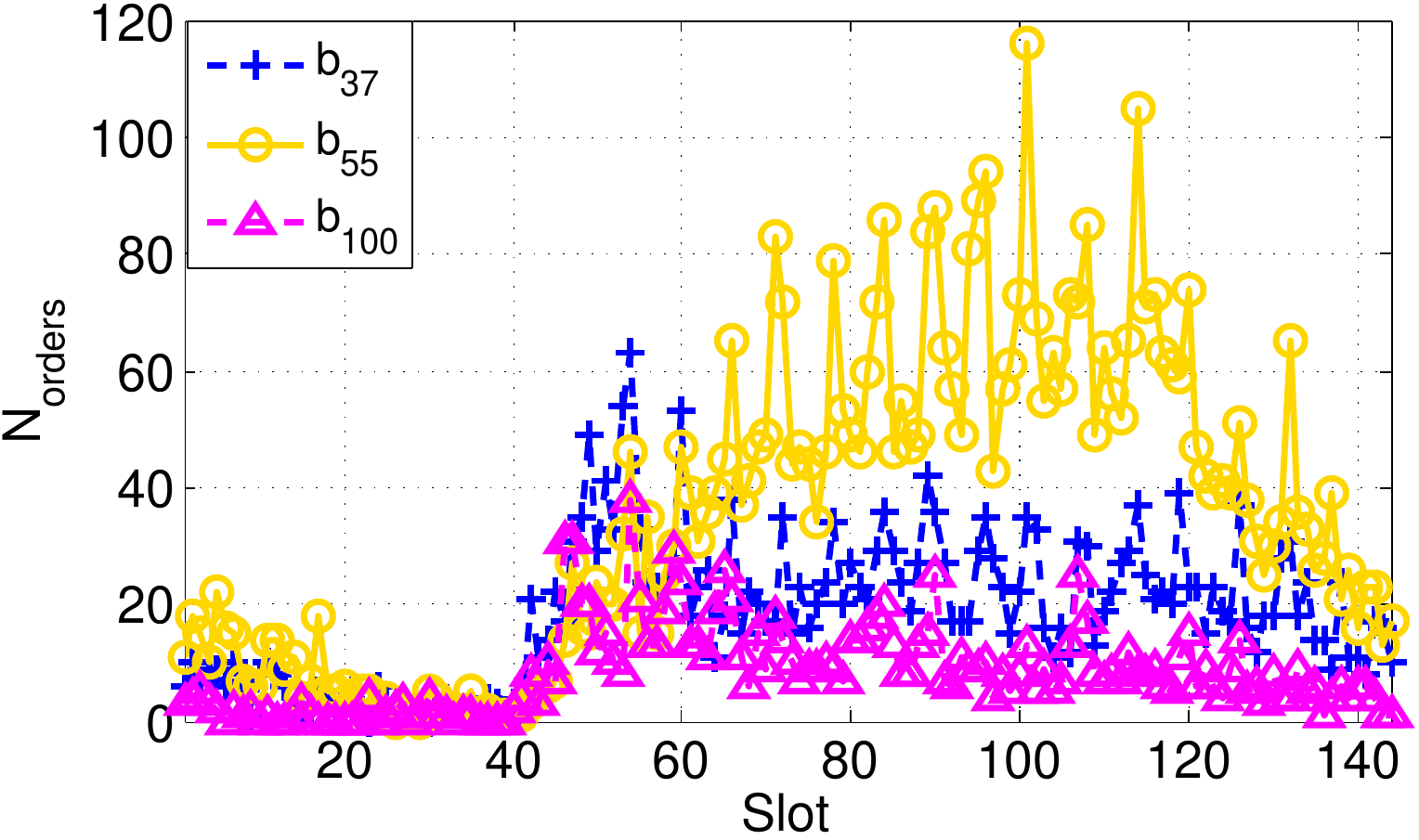}
\caption{The order number of one block in each slot.}
\label{fig_norder}
\end{minipage}
\end{figure*}

Fig. \ref{fig_ptnum} shows the package matching rate (MR) under different waiting time. In the first step of each package, the platform should dispatch a taxi to fetch the package. In this paper, the nearest taxi within the same block and slot of the package defaults to be dispatched. If a package is fetched by a taxi, we denote the package as matched. MR is the ratio between the number of matched packages and the number of all packages. The route planning in our proposed algorithm is started at the departure time of each package. Between the generation time $u_{genT}$ and the departure time $u_{depT}$ of a package, there is a period of waiting time. In this experiment, we show the relationship between the MR and the waiting time of packages. We can see that, when the generation time is $3\mathrm{:}00$, the waiting time for each package can be up to one hour to reach $90\%$ MR. When the generation time is $12\mathrm{:}00$ or $20\mathrm{:}00$, $10$ minutes is enough. It means there are enough taxis for package delivery in the daytime, which further implies that the MR is directly impacted by the taxi density in the target area.

Fig. \ref{fig_norder} illustrates the candidate order number ($N_{orders}$) in each departure slot within a block. $b_{i}$ denotes the $i^{th}$ block. In this experiment, we select different blocks to evaluate, including the center ($b_{55}$), the border ($b_{100}$), and the randomly selected block ($b_{37}$) of the target area. We can see that the candidate orders are the most in the area center and least at the border. The average order number of rush hours in these three blocks are around $10$, $23$, and $60$. This observation result explains the time-cost phenomenon of the PSP algorithm.

%--------------------------------------------------------------------------------------------------------------------------------------
\section{Related Work}
\label{sec_relat}

\subsection{Crowdsourced package delivery}

City logistics has recently appropriated an exponential growth of publications \cite{dolati2018systematic,cleophas2018collaborative}. Many novel ideas have been proposed to speed up the city-wide package delivery, such as multi-model express \cite{kim1999multimodal}, UAV-based (Unmanned Aerial Vehicle) delivery \cite{agha2014health}, and crowdsourced delivery, which is the focus of our paper. In the crowdsourced delivery, the packages can be delivered with common people instead of professionals, resulting in economic, social and environmental benefits. There are two kinds of directions in the crowdsourced delivery, dedicated delivery, and ridesharing based delivery. 

Dedicated delivery indicates that the participants are recruited to deliver packages deliberately. There are no concurrent tasks during a delivery. For example, Sadilek et al. \cite{sadilek2013crowdphysics} recruited a group of Twitter users, asking one person to pass the assigned package to another Twitter user that happened to be nearby. Within this direction, the matching between participants and packages are similar to the task allocation problem, where multiple research papers have been developed \cite{chen2018survey,gong2018task,guo2018task}. 

Ridesharing based delivery means that the delivery operations are carried out by using excess capacity on journeys that are already taking place. The ridesharing based passenger transporting is also a popular development direction for the online taxi-taking platforms, such as DiDi \cite{didi} and Uber \cite{uber}. However, most of the ridesharing systems are based on isomorphic tasks, either passenger \cite{ma2013t} or packages \cite{rouges2014crowdsourcing,setzke2017matching}, while few papers focus on the isomerous ridesharing (the package-passenger ridesharing). Walmart proposes to make use of its in-store customers to deliver goods to its online customers on their way home from the store \cite{Arslan2016}. Liu et al. \cite{Liu2018} propose to investigates the participant of urban taxis to support on demand take-out food delivery.

The most related work with this paper is \cite{Chen2017} and \cite{wangridesharing}. Chen et al. \cite{Chen2017} exploit relays of taxis with passengers to help transport package collectively, without degrading the quality of passenger solutions. Consignment warehouses are deployed along the roads as relay points, and one package could be delivered under the cooperation of several taxis. This method reduces the delivery cost while increasing the storage cost, as many consignment warehouses are needed to cover the whole city. In our paper, no consignment warehouses are needed, and the package delivery is non-stop. Wang et al. \cite{wangridesharing} propose a ridesharing based package delivery solution. When the pickup location of the package is close to the vehicle's source, and the drop off location is close to the vehicle's destination, the vehicle is assumed to deliver the package in a ridesharing way. An online algorithm is designed to match the package delivery tasks to the submitted vehicles, aiming to maximize the utility of the ridesharing solution provider. In this paper, only one hop of ridesharing is considered, while in our paper, we discuss the ridesharing routs consisting of multiple hops of passengers to successfully deliver a package.

\subsection{Passenger order prediction and dispatch}

Passenger order prediction is essential for the taxi deployment and scheduling to enabling intelligent transportation systems in a smart city. Two parameters are essential within the order prediction, i.e., order number and passenger flow. Order number indicates the number of the taxi-taking orders submitted per unit time and per unit block. To improve the prediction accuracy, many recent works discuss and combine the multi-source data fusion and deep learning technologies \cite{tong2017simpler,yao2018deep}. The crow flow prediction is deeply studied in other areas, for example, Zhang et al. \cite{zhang2017deep} propose a deep-learning approach to collectively forecast the indemand and outdemand of crowds in each block of a city. However, this method can not be used in the passenger flow prediction directly. Passenger flow indicates to predict the direction of passengers appeared in each block within each time slot. Passenger flow prediction is important for optimizing the order dispatch strategies in the online taxi-taking applications. Zhang et al. \cite{Zhang2017taxi} propose to predict destinations of a user once the taxi-booking application is started by employing the Bayesian framework to model the distribution of a user's destination based on his/her travel histories. However, only the personal destination is predicted in Zhang's paper, and the general flow prediction is not discussed. Note that, our delivery solution is orthogonal with all existing passenger flow prediction methods. The combination of our route planning algorithm and more accurate flow prediction can further improve the success rate of package delivery.

Passenger order dispatch is another important research area for the online taxi-taking platform. Tong et al. \cite{tong2016,tong2017flexible} discussed the matching problem between passengers and taxis in perspectives. Zhang et al. \cite{Zhang2017taxi} propose a novel system to optimally dispatch taxis to serve multiple bookings, aiming to maximize the global success rate. Xu et al. \cite{xu2018large} designs an algorithm to provide a more efficient way to optimize the resource utilization and user experience in a global and more farsighted view, instead of focusing on immediate customer satisfaction. The goal of all these papers is to optimize the platform's efficiency to serve passengers more satisfactorily. In our paper, the order dispatch algorithm is altered to serve the package delivery and passenger transporting synchronously. Thus, it may sacrifices the platform efficiency and user experience to some degree compared with the traditional dispatch methods. For example, a taxi may need to go further to take an optimal passenger order along with our planned route, compare with the traditional order dispatch algorithms. However, We believe that the reward brings by the package delivery can make up the induced problem. There is a lot of details within this scenario can be discussed.

%--------------------------------------------------------------------------------------------------------------------------------------
\section{Conclusions and Future Works}
\label{sec_conclusion}

In this paper, we propose a non-stop package delivery solution. One taxi is responsible for one package delivery, while the passenger-transporting process of this taxi is not interrupted. We formulate the non-stop package delivery problem (NPD) of route-probability maximization, under the constraint of package delivery delay. The route probability is calculated according to the predicted passenger flow probabilities. We tackle the NPD problem with a two-phase solution, named \textbf{PPtaxi}. In the passenger order prediction phase, we mine the historical order data to calculate the probability of each passenger flow. In the package route planning phase, both the computation efficiency and solution effectiveness are considered. At last, we evaluate our solution using a real-world dataset and compare it with multiple benchmarks. 

Although the evaluation results show the effectiveness and efficiency of our solution, it can be further extended from the following aspects to make it more practical to use. First, more precise prediction methods could be considered when mining the passenger regulars, as we take the passenger prediction into account in our solution. The machine learning methods combined with multiple dimensional data sources (such as weather information and traffic information) may improve the prediction precision. Second, as the packages appear dynamically, we plan a delivery route for each package when it appears in this paper. If many packages appear together, our route planning strategy may induce resource congestion between packages. For example, the planned routes for two different packages (different at pick up location or drop off location) may need the same hop of passenger order. Hence, a congestion-aware route planning strategy may further improve the success rate of our package delivery solution. Third, the package delivery process can be considered jointly with the passenger delivery process in the real world to met the requirements of both packages and passengers. The co-dispatch of the passengers and packages within the online taxi-taking platform is an interesting research point. We leave these topics for our future work. 

\bibliographystyle{IEEEtran}
\bibliography{yueyue}



%\vfill


\end{document}